\documentclass[11pt]{article}

\usepackage{url,amsmath,amsthm,amscd,amsfonts,graphicx,psfrag,marginnote}
\usepackage[all]{xy}
\usepackage{fancyhdr}
\usepackage{wrapfig}

\newtheoremstyle{stuff}{\topsep}{\topsep}%
     {\itshape}
     {}
     {\bfseries}
     {}
     {.5em}
     {\thmnote{#3}}

\theoremstyle{plain}
\newtheorem{thm}{Theorem}

\newtheorem{lem}[thm]{Lemma}

\newtheorem{example}{Example}
\newtheorem{conjecture}{Conjecture}

\theoremstyle{remark}
\newtheorem*{rem}{Remark}

\theoremstyle{definition}
\newtheorem{defn}{Definition}
\newtheorem*{question}{Question}

\theoremstyle{stuff}

\numberwithin{equation}{section}
 
\usepackage[scale=0.7]{geometry}

\usepackage{sectsty}
\sectionfont{\center \sc \large}
\subsectionfont{\bf \large}
\subsubsectionfont{\it \large}

\usepackage{setspace}

\newcommand{\bea}{\begin{eqnarray}}
\newcommand{\eea}{\end{eqnarray}}
\newcommand{\be}{\begin{equation}}
\newcommand{\ee}{\end{equation}}

\begin{document}

\pagestyle{empty}

\begin{flushright}
\begin{tabular}{l}
CALT-68-2836 \\
\end{tabular}
\end{flushright}

{\noindent \Large \bf \fontfamily{pag}\selectfont Topological recursion and mirror curves}\\

\vspace*{0.2cm}
\noindent \rule{\linewidth}{0.5mm}

\vspace*{0.8cm}

{\noindent \fontfamily{pag}\selectfont  Vincent Bouchard}\\[0.2em]
{\small {\it \indent Department of Mathematical and Statistical Sciences\\
\indent University of Alberta\\
\indent 632 CAB, Edmonton, Alberta T6G 2G1\\
\indent Canada}\\
\indent \url{vincent@math.ualberta.ca}}\\[0.3em]

{\noindent \fontfamily{pag}\selectfont  Piotr Su{\l}kowski}\\[0.2em]
{\small {\it \indent California Institute of Technology\\
\indent Pasadena, CA 91125\\
\indent USA\\[0.2em]
\indent Faculty of Physics, University of Warsaw\\
\indent ul. Ho{\.z}a 69, 00-681 Warsaw\\
\indent Poland}\\
\indent \url{psulkows@theory.caltech.edu}}\\[0.3em]

\vspace*{0.8cm}

\hspace*{1cm}
\parbox{11.5cm}{{\sc Abstract:} We study the constant contributions to the free energies obtained through the topological recursion  applied to the complex curves mirror to toric Calabi-Yau threefolds. We show that the recursion reproduces precisely the 
corresponding Gromov-Witten invariants, which can be encoded in powers of the MacMahon function. 
As a result, we extend the scope of the ``remodeling conjecture'' to the full free energies, including the constant contributions. In the process we study how the pair of pants decomposition of the mirror curves plays an important role in the topological recursion. We also show that the free energies are not, strictly speaking, symplectic invariants, and that the recursive construction of the free energies does not commute with certain limits of mirror curves.}

\pagebreak

\pagestyle{fancy}
\pagenumbering{roman}

\tableofcontents 


\pagebreak
\pagenumbering{arabic}

\section{Introduction}

The ``remodeling conjecture'' \cite{Bouchard:2009, Marino:2008} asserts that the generating functions of Gromov-Witten invariants of a toric Calabi-Yau threefold $X$ are completely determined in terms of a topological recursion. The particular recursion is the Eynard-Orantin topological recursion \cite{Eynard:2007, Eynard:2008} applied to the complex curve $\Sigma$ mirror to $X$.

\subsection{Constant terms}

The simplest Gromov-Witten invariants are those involving constant maps to the target space. They are encoded by the leading constant term $N_{g,0}$ in the Gromov-Witten generating functions $F_g$. It is known since the work of \cite{Faber:2000, Gopakumar:1998, Marino:1999} that for $g \geq 2$, the constant terms are given by:
\begin{equation}\label{eq:FP1}
N_{g,0} =  \frac{1}{2} (-1)^g \chi(X) \frac{|B_{2g}| |B_{2g-2}|}{2g(2g-2)(2g-2)!},
\end{equation}
where $\chi(X)$ is the topological Euler characteristic of $X$. 

In this paper we ask the following question: is the remodeling conjecture true for the full free energies $F_g$, including constant terms, or just for the ``reduced free energies'' without the constant terms? In other words, are the $F_g$ constructed from the Eynard-Orantin recursion applied to the complex curve mirror to $X$ reproducing the ``right'' constant terms as in \eqref{eq:FP1}? The fate of constant terms is notoriously subtle, as is well known for instance from the DT/GW correspondence \cite{Maulik:2006, Maulik:2006ii}. 

As far as we are aware, constant terms have not been studied yet from the point of view of the remodeling conjecture. The main reason is that most of the checks and proofs of the conjecture have been done by comparing with the topological vertex on the Gromov-Witten side, which computes only the reduced Gromov-Witten theory. Hence not much has been said about constant maps.

Apart from clarifying the remodeling conjecture, constant terms are interesting for many reasons. For instance, over the years matrix models have been constructed which encode Gromov-Witten partition functions $Z$ of various toric geometries \cite{Aganagic:2004,Eynard:2008ii,Eynard:2010,Eynard:2010ii,SW-matrix,Marino:2004,matrix2star,Ooguri:2010}. The spectral curves of these matrix models give the corresponding mirror curves. However, the constant part of the matrix models \emph{do not} generally give the right power of the MacMahon function to recover the contributions from constant maps in Gromov-Witten theory. Thus one can ask: Are the $F_g$ constructed through the Eynard-Orantin recursion for these spectral curves giving the Gromov-Witten constant terms, or the constant terms of the corresponding matrix models? \emph{A priori}, one would think that they should give the constant terms of the matrix models, since the $F_g$ obtained from the recursion are supposed to be the free energies of the corresponding matrix models. But are the loop equations, to which the recursion is a solution, really aware of the overall matrix model normalization?

This question occurs already in the case of the resolved conifold. Here, various matrix models are known (see for instance \cite{Aganagic:2004,Eynard:2008ii,Marino:2004, Ooguri:2010}). They all give spectral curves that are ``symplectically equivalent'', hence should all produce the same free energies through the Eynard-Orantin recursion. However, these matrix models have different normalizations, \emph{i.e.} different powers of the MacMahon functions for the constant terms. More generally, matrix models built from the topological vertex formalism for arbitrary toric threefolds \cite{Eynard:2010,SW-matrix,matrix2star} \emph{by construction} do not involve any factors of MacMahon function, while matrix models constructed by other means \cite{Aganagic:2004,Marino:2004,Ooguri:2010} \emph{by construction} involve such factors. Yet, they have been shown to determine the same mirror curves \cite{Eynard:2010ii,Ooguri:2010}. How can that be? Given a spectral curve (or two symplectically equivalent curves), the $F_g$ are uniquely constructed by the recursion, including constant terms. How can they reproduce the free energies of matrix models with different constant terms?

What seems to be happening is the following. The Eynard-Orantin recursion is a solution to the loop equations of matrix models. The loop equations however do not ``know'' about constant terms; two matrix models differing only by overall normalization should give the same loop equations, hence ``symplectically equivalent'' spectral curves. The recursion however \emph{does} compute constant terms. Therefore, those are not necessarily the constant terms of the corresponding matrix models; the recursion replaces the constant terms of the matrix models by its own preferred constant terms. 

And what we argue in this paper is that, magically, \emph{the constant terms computed through the recursion are precisely those of Gromov-Witten theory}, as given in \eqref{eq:FP1}! In other words, the recursion knows the right constant terms. As a result, we assert that the remodeling conjecture holds for the full free energies, including constant terms.

What we show is that for the simplest Calabi-Yau threefold $X = \mathbb{C}^3$, the $F_g$ are given by the Faber-Pandharipande formula \eqref{eq:FP1} with $\chi(X) = 1$. Then, we argue that for any toric Calabi-Yau threefold $X$, the constant part of the $F_g$ will be given by $\chi(X)$ times the $F_g$ of $\mathbb{C}^3$, thus recovering \eqref{eq:FP1} for all toric Calabi-Yau threefolds.

It would be nice to understand better, from a matrix model point of view, why the Eynard-Orantin recursion does produce the right constant terms as in Gromov-Witten theory. This remains quite mysterious, and should be addressed in future work.

As an aside we study, along the way, in more detail the notion of ``symplectic invariance'' of the free energies $F_g$, as well as their behavior under certain limits of mirror curves. It turns out that the $F_g$ are not, strictly speaking, symplectic invariants; we  show that symplectic transformations that change the number of ramification points do not leave the $F_g$ invariant. In the context of the remodeling conjecture, what happens is that there exists ``bad choices'' of framing for which the topological recursion does not produce the correct free energies. In fact, representations of the mirror curves quite often encountered in literature, such as (\ref{curveC3}) for $\mathbb{C}^3$ and (\ref{curveConifold}) for the conifold, correspond to such pathological choices of framing. 
We also demonstrate that the recursive construction of the $F_g$ does not commute with certain limits of mirror curves.
We plan to clarify these issues further.

\subsection{Pair of pants decomposition}

A related question that we address in this paper is whether the pair of pants decomposition of the mirror curves plays a role in the recursion. Let $\Sigma$ be the curve mirror to a toric Calabi-Yau threefold $X$. Any such $\Sigma$ has a pair of pants decomposition. What we show is that this pair of pants decomposition is in one-to-one correspondence with the $\mathbb{C}^3$ patches decomposition of the mirror toric threefold $X$. And just as the $\mathbb{C}^3$ patches decomposition plays a crucial role in Gromov-Witten theory through the topological vertex formalism \cite{Aganagic:2005, Li:2009, Maulik:2011}, the pair of pants decomposition also plays an important role on the mirror B-model side through the Eynard-Orantin recursion. 

It turns out that the residue process at the ramification points built in at each level of the Eynard-Orantin recursion implements directly this pair of pants decomposition. Each pair of pants has its corresponding ramification point, and summing over ramification points means that we are summing over contributions for each pair of pants, at each level of the recursion. In particular, for a given $X$, we argue that each pair of pants contribute precisely the same amount to the constant part of its $F_g$, namely, the contribution coming from constant maps to $\mathbb{C}^3$. In the case of the resolved conifold, we prove that the two pairs of pants contribute precisely the same amount to the full $F_g$, not just for constant terms. It would be very nice to see how to understand the pair of pants decomposition of the free energies for general toric Calabi-Yau threefolds $X$ beyond the constant terms.

As a result, the residue process seems to be a close analog to the localization procedure in Gromov-Witten theory of toric manifolds. It would also very interesting to make this analogy more precise.

\subsection{Outline}

In section 2, we review the basics of Gromov-Witten theory of toric Calabi-Yau threefolds and the construction of the corresponding mirror curves. We also review the Eynard-Orantin recursion, and state the remodeling conjecture. Section 3 is devoted to the analysis of constant terms. We first explain the role of ramification points in the recursion, and how they are mirror to the torus fixed points of the toric Calabi-Yau threefold. We conjecture (and check to relatively high genus) that the $F_g$ constructed from the mirror curve to $\mathbb{C}^3$ are equal to the Gromov-Witten result \eqref{eq:FP1} with $\chi(X) = 1$. Then we argue that for a general toric threefold $X$, each pair of pants contributes one copy of $F_g^{\mathbb{C}^3}$, thus recovering the full Faber-Pandharipande formula \eqref{eq:FP1} for general $X$. We prove this result explicitly for the case of the resolved conifold. 

Finally, in section 4 we discuss some of the issues that were encountered during the analysis of section 3. In particular, we discuss the notion of symplectic invariance of the $F_g$. We show that under some symplectic transformations that do not preserve the number of ramification points, the $F_g$ are not invariant. This analysis also gives a new meaning to the framing of mirror curves, as some sort of ``regularization parameter''. We also discuss the relation between our results for constant terms and the limit theorem obtained by Eynard and Orantin in \cite{Eynard:2007}.

\subsection*{Acknowledgments}
We would like to thank Cedric Berndt, Bertrand Eynard, Sergei Gukov, Amir Kashani-Poor and Olivier Marchal for enjoyable discussions. The research of V.B. is supported by a University of Alberta startup grant and an NSERC Discovery grant. The research of P.S. is supported by the DOE grant DE-FG03-92ER40701FG-02 and the European Commission under the Marie-Curie International Outgoing Fellowship Programme.

\section{Topological recursion and the remodeling conjecture}

In this section we review 
the remodeling conjecture, and the definition of the topological recursion that we will be interested in.

\subsection{Gromov-Witten theory of toric Calabi-Yau threefolds}

\label{s:GW}

In this paper we will be interested in Gromov-Witten theory of toric Calabi-Yau threefolds.

\subsubsection{Toric Calabi-Yau threefolds}

Let $X$ be a toric threefold. It can be written as
\begin{equation}
X = \frac{\mathbb{C}^{3+k} \setminus S}{(\mathbb{C}^*)^k},
\end{equation}
where the $k$ $\mathbb{C}^*$ actions are given by
\begin{equation}
(\mathbb{C}^*)^i: (z_1, \ldots, z_{3+k}) \mapsto (\lambda^{T_{i1}} z_1, \ldots, \lambda^{T_{i(3+k)}} z_{3+k}), \quad i=1,\ldots,k, \quad \lambda \in \mathbb{C}^*,
\end{equation}
where the $T_{ij} \in \mathbb{Z}$ and $S$ is a subset which is fixed by a continuous subgroup of $(\mathbb{C}^*)^k$. In this notation, the geometry is determined by the $k$ vectors $T_i$, sometimes known as \emph{toric charges}.

It is well known that $X$ is Calabi-Yau if and only if
\begin{equation}
\sum_{a=1}^{3+k} T_{i a} = 0, \qquad i=1,\ldots,k.
\end{equation}
As a result, every toric Calabi-Yau threefold is noncompact. Toric Calabi-Yau threefolds can be represented by trivalent graphs, known as \emph{toric diagrams}; the edges represent the torus-invariant curves in $X$, while the vertices correspond to the torus fixed points.

\begin{example}
Here a few examples of toric Calabi-Yau threefolds and their toric charges. The toric diagrams are shown in fig. \ref{fig-toric}.
\begin{itemize}
\item $X = \mathbb{C}^3$, which is the simplest toric Calabi-Yau threefold. 
\item $X = \mathcal{O}_{\mathbb{P}^1} (-1) \oplus  \mathcal{O}_{\mathbb{P}^1} (-1) $, known as the \emph{resolved conifold}. Its toric charge is $T = (-1, -1, 1, 1)$.
\item $X = \mathcal{O}_{\mathbb{P}^2}(-3)$, known as \emph{local $\mathbb{P}^2$}, which has toric charge $T = (-3,1,1,1)$.
\end{itemize} 
\end{example}

\begin{figure}[htb]
\begin{center}
\includegraphics[width=\textwidth]{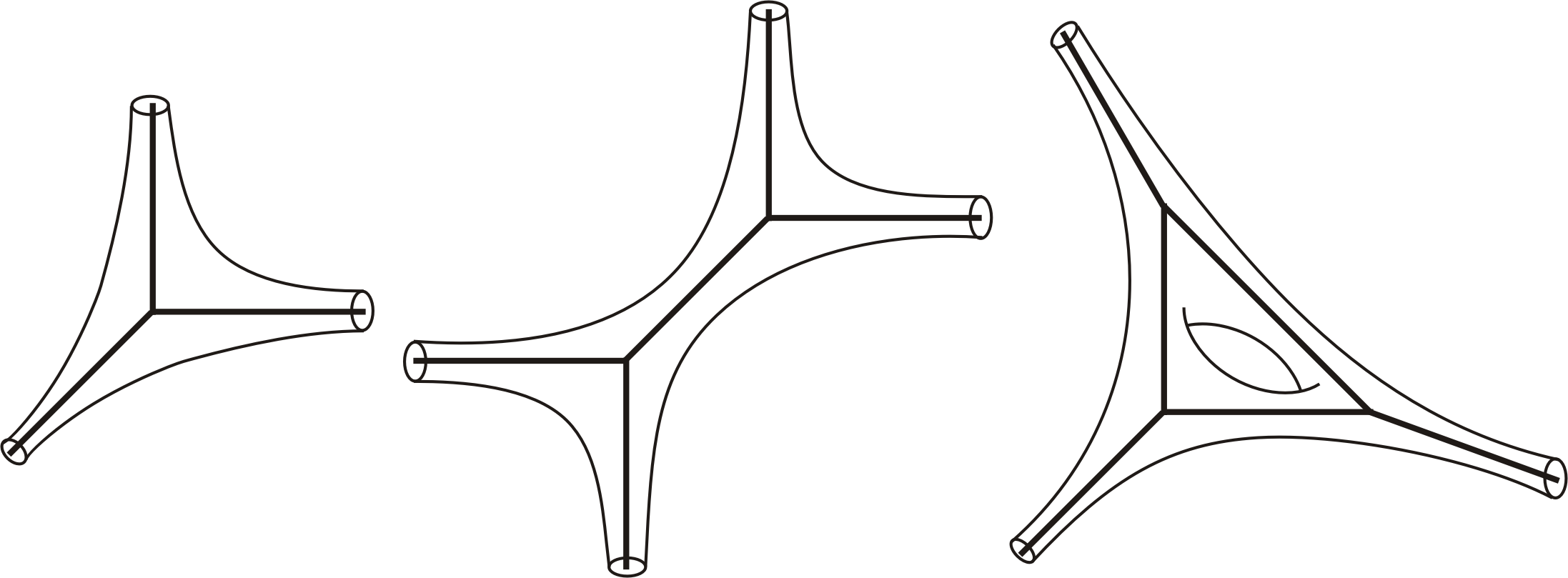} 
\begin{quote}
\caption{\emph{Toric diagrams for $\mathbb{C}^3$, resolved conifold, and local $\mathbb{P}^2$. The corresponding mirror curves arise by thickening the edges of toric diagrams. The trivalent vertices in the toric diagram represent torus fixed points. 
As we demonstrate in section \ref{s:bpoints}, these torus fixed points are in one-to-one correspondence with ramification points of the corresponding mirror curve (for generic choice of framing). Hence the A-model computation based on gluing of topological vertices is mirrored by the pair of pants decomposition of the topological recursion computation in the B-model.}   } \label{fig-toric}
\end{quote}
\end{center}
\end{figure}

\subsubsection{Gromov-Witten theory}

Consider the moduli space $\mathcal{M}_{g,n}(X, \beta)$ of stable maps $f : S_g \to X$ from $n$-pointed genus $g$ Riemann surfaces $S_g$, with homology class $f_*[S_g] = \beta \in H_2(X, \mathbb{Z})$. To define Gromov-Witten invariants, we consider the Deligne-Mumford compactification $\overline{\mathcal{M}}_{g,n}(X, \beta)$ of the moduli space of stable maps, and then construct its virtual fundamental class $[ \overline{\mathcal{M}}_{g,n}(X, \beta)]_{\text{virt}}$. In this paper we focus on the case with no marked points, $\overline{\mathcal{M}}_{g}(X. \beta):= \overline{\mathcal{M}}_{g,0}(X. \beta)$. 

Roughly speaking, Gromov-Witten invariants are then defined as integrals of appropriate cohomology classes over this virtual fundamental class:
\begin{equation}
\langle  \cdots \rangle_{g,\beta} = \int_{[\overline{\mathcal{M}}_{g}(X. \beta)]_{\text{virt}}} ( \cdots ).
\end{equation}
For $X$ a Calabi-Yau threefold, $\overline{\mathcal{M}}_{g}(X, \beta)$ has virtual dimension $0$. Thus we can define the following Gromov-Witten invariants (for $g \geq 2$):
\begin{equation}
N_{g,\beta} = \langle \rangle_{g,\beta} = \int_{[\overline{\mathcal{M}}_{g}(X, \beta)]_{\text{virt}}} 1 = \text{deg}[\overline{\mathcal{M}}_{g}(X, \beta)]^{\text{virt}} .
\end{equation}
Those are the invariants that we will be interested in.

As usual, we form generating functions for these invariants:
\begin{align}
F_g =& \sum_{\beta \in H_2(X,\mathbb{Z})} N_{g,\beta} Q^\beta \qquad \text{(genus $g$ free energies)}, \\
F =& \sum_{g=0}^\infty \lambda^{2g-2} F_g \qquad \text{(free energy)}, \label{Flambda} \\
Z=&  \text{exp}( F) \qquad \text{(partition function)}.
\end{align}
 
It is also customary to separate the constant maps $\beta=0$ from the non-constant ones. We define the \emph{reduced free energies} and \emph{reduced partition function} as
\begin{align}
F_g^{\beta \neq 0} =& \sum_{\substack{\beta \in H_2(X,\mathbb{Z})\\ \beta \neq 0}} N_{g,\beta} Q^\beta,    \label{Fg-beta}  \\
F^{\beta \neq 0} =& \sum_{g=0}^\infty \lambda^{2g-2} F_g^{\beta \neq 0},\\
Z^{\beta \neq 0}=& \text{exp}( F^{\beta \neq 0}).
\end{align}
Then we have that
\begin{equation}
Z = Z^{\beta \neq 0} Z^{\beta = 0},   \label{Ztotal}
\end{equation}
where 
\begin{equation}\label{eq:constant}
Z^{\beta=0} = \text{exp}\left(\sum_{g=0}^\infty \lambda^{2g-2} N_{g,0}  \right)
\end{equation}
involves only the constant map contributions.

\subsubsection{What is known}

\label{s:known}

Gromov-Witten theory of toric Calabi-Yau threefolds has been studied extensively. For any toric Calabi-Yau threefold $X$, the reduced partition function $Z^{\beta \neq 0}$ can be computed using the so-called \emph{topological vertex} \cite{Aganagic:2005, Li:2009, Maulik:2011}. The topological vertex however does not say anything about the constant contributions $Z^{\beta=0}$.

Fortunately the constant contributions can be computed independently \cite{Faber:2000, Gopakumar:1998, Marino:1999}. For $g \geq 2$, one gets that:
\begin{equation}\label{eq:FP}
N_{g,0} =  \frac{1}{2} (-1)^g \chi(X) \frac{|B_{2g}| |B_{2g-2}|}{2g(2g-2)(2g-2)!},
\end{equation}
where $\chi(X)$ is the topological Euler characteristic of $X$. For $X$ a toric Calabi-Yau threefold,  it is easy to show that $\chi(X)$ is equal to the number of torus fixed points (that is, the number of vertices in the toric diagram of $X$). In particular, for the simplest toric Calabi-Yau threefold $X = \mathbb{C}^3$, which has $\chi(X) = 1$, we have
\begin{equation}\label{eq:Fgc3}
N_{g,0}^{\mathbb{C}^3} = \frac{ 1}{2} (-1)^g \frac{|B_{2g}| |B_{2g-2}|}{2g(2g-2)(2g-2)!}. 
\end{equation}
For a general toric Calabi-Yau threefold, we thus obtain
\begin{equation}\label{eq:Fgrelation}
N_{g,0}^X = \chi(X) N_{g,0}^{\mathbb{C}^3}.
\end{equation}

The result \eqref{eq:FP} can be neatly rewritten as  \cite{Behrend:2005, Li:2006,Maulik:2006, Maulik:2006ii}
\begin{equation}
Z^{\beta=0} = M(q)^{\frac{1}{2}\chi(X)},
\end{equation}
where $M(q)$ is the MacMahon function
\begin{equation}\label{eq:macmah}
M(q) = \prod_{k=1}^\infty \left(1-q^k \right)^{-k},
\end{equation}
with $q = \mathrm{e}^{\mathrm{i} \lambda}$.

\subsection{Mirror symmetry}

\label{s:mirror}

The free energy $F$ of A-model topological string theory on a toric Calabi-Yau threefold $X$ generates the Gromov-Witten invariants as in the previous subsection. It is however often useful to study the free energy $F$ from a dual point of view. Mirror symmetry relates the A-model free energy to the B-model topological string free energy, which we also denote by $F$, through the mirror map. The remodeling conjecture proposes a recursive formula to compute the B-model genus $g$ free energies $F_g$. Let us start by explaining what the mirror B-model theory looks like.

\subsubsection{The mirror curve}

Here we construct the mirror geometry following Hori-Vafa \cite{Hori:2000}. We are interested in the mirror B-model theory to the A-model on a toric Calabi-Yau threefold $X$. The mirror theory is generally formulated as a Landau-Ginzburg theory. However, it was shown that it essentially reduces to the geometry of a complex curve, known as the \emph{mirror curve}. In fact, the remodeling conjecture is formulated entirely in terms of this mirror curve. Therefore here we will define the mirror theory directly in terms of a mirror curve, skipping the Landau-Ginzburg step.

Let $X$ be a toric Calabi-Yau threefold defined by some toric charges $T_i$, $i=1,\ldots,k$. Intuitively, the mirror curve can be obtained by ``fattening'' the toric diagram, as shown in fig. \ref{fig-toric}. Mathematically, it is defined as follows.

\begin{defn}
The \emph{mirror curve} of a toric Calabi-Yau threefold $X$ defined by the toric charges $T_i$, $i=1,\ldots,k$, is given by the family of curves:
\begin{equation}\label{eq:mc}
\Sigma = \left\{1 + x_2 + \ldots +x_{3+k} = 0 ~\Big |~ r_i  = \prod_{m=2}^{3+k} x_m^{T_{im}}, ~ i=1,\ldots,k \right\},
\end{equation}
where $x_m \in \mathbb{C}^*$ for $m=2, \ldots, 3+k$.
\end{defn}
Note that this defines a family of curves in $(\mathbb{C}^*)^2$, where the family is parameterized by the $r_i$. For instance, we can choose $x:=x_2$ and $y:=x_3$ to be our $\mathbb{C}^*$ variables, and eliminate the other $x_m$'s using the definition of the parameters $r_i$. The mirror curve $\Sigma$ defines a Riemann surface with punctures (which we also denote by $\Sigma$), and $x$ and $y$ are holomorphic functions on $\Sigma$.

\begin{example}
Let $X=\mathbb{C}^3$. According to \eqref{eq:mc}, the mirror curve is
\begin{equation}
\Sigma = \left \{ 1 + x + y = 0 \right \} \subset (\mathbb{C}^*)^2.     \label{curveC3}
\end{equation}
The associated Riemann surface has genus $0$ and three punctures (shown in figure \ref{fig-toric}).
\end{example}

\begin{example}
Let $X$ be the resolved conifold, $X= \mathcal{O}_{\mathbb{P}^1} (-1) \oplus  \mathcal{O}_{\mathbb{P}^1} (-1) $, with toric charge $T = (-1, -1, 1, 1)$. The mirror curve is
\begin{equation}
\Sigma = \left \{ 1 +x +y + r x y^{-1} = 0 \right \} \subset (\mathbb{C}^*)^2 .        \label{curveConifold}
\end{equation}
The associated Riemann surface has genus $0$ and four punctures (shown in figure \ref{fig-toric}). Note that $\Sigma$ is a family of curves, parameterized by $r$.
\end{example}

\begin{example}
Consider local $\mathbb{P}^2$, $X = \mathcal{O}(-3)_{\mathbb{P}^2}$, with toric charge $T=(-3,1,1,1)$. The mirror curve is
\begin{equation}
\Sigma = \left \{ 1 +x +y + r x^{-1} y^{-1} = 0 \right \}.    \label{curveP2}
\end{equation}
The associated Riemann surface has genus $1$ and three punctures (shown in figure \ref{fig-toric}).
\end{example}

\subsubsection{Framing}

It turns out that the mirror curve can also depend on another parameter, known as ``framing''. In Gromov-Witten theory, the choice of framing appears because one needs to fix the torus weights in order to apply localization under the torus action. It can also be understood from the point of view of large $N$ dualities with Chern-Simons theory. For us, the choice of framing appears as a reparameterization of the curve.

\begin{defn}
Let $\Sigma$ be a mirror curve, given by the locus $\{H(x,y)=0\} \subset (\mathbb{C}^*)^2$. We define the associated \emph{framed curve} $\Sigma_f$ by the $\mathbb{C}^*$ reparameterization
\begin{equation}
Y = y, \qquad X = x  y^f.   \label{def-framing}
\end{equation}
The framed curve is given by the locus $\{H_f(X,Y)=0\} \subset (\mathbb{C}^*)^2$, and it depends on a new parameter $f$, known as \emph{framing}.
\end{defn}

\begin{rem}
Note that framing may seem a bit \emph{ad hoc}, but as we will see it is rather crucial.
\end{rem}

\begin{example}
The mirror curve to $\mathbb{C}^3$ is \eqref{curveC3}. The framed curve $\Sigma_f$ is then
\begin{equation}
\Sigma_f = \left \{ 1 + X Y^{-f} + Y = 0 \right \} \subset (\mathbb{C}^*)^2,
\end{equation}
or equivalently
\begin{equation}\label{eq:fmc3}
\Sigma_f = \left \{ X+ Y^f + Y^{f+1} = 0 \right \} \subset (\mathbb{C}^*)^2,
\end{equation}
\end{example}

\subsection{Topological recursion}

\label{s:recursion}

The remodeling conjecture is based on a particular topological recursion, known as Eynard-Orantin recursion \cite{Eynard:2007, Eynard:2008}. In this section we define the Eynard-Orantin topological recursion.

\subsubsection{Ingredients}

We start with a smooth affine plane curve
\begin{equation}
C = \{ H(x,y) = 0 \} \in \mathbb{C}^2.
\end{equation}
It defines a non-compact Riemann surface, which we also denote by $C$. $x,y: C \to \mathbb{C}$ are holomorphic functions on $C$. As usual, $C$ can be compactified to $\hat{C}$ by adding points at infinity, and $x,y: \hat{C} \to \mathbb{C}_\infty$ become meromorphic functions on the compact Riemann surface $\hat{C}$.

We assume that the map $x: C \to \mathbb{C}$ has only simple ramification points. Let $\{a_1, \ldots, a_n\} \in C$ be the set of simple ramification points of $x$. Locally, at each $a_\lambda$, $\lambda=1, \ldots, n$, the map is a double-sheeted covering, hence we have a deck transformation map
\begin{equation}
s_\lambda: U_\lambda \to U_\lambda
\end{equation}
which is defined locally in a neighborhood $U_\lambda$ of $a_\lambda$. The deck transformation map means that
\begin{equation}
x(t) = x(s_\lambda(t))
\end{equation}
for some local coordinate $t$ near $a_\lambda$.

The type of objects that we will be interested in are meromorphic symmetric differentials. Let $\mathcal{M}^1_C$ be the sheaf of meromorphic one-forms on $C$. A degree $n$ meromorphic differential $W_n(p_1, \ldots, p_n)$ is a section
\begin{equation}
W_n(p_1, \ldots, p_n) \in H^0 ( C^n,  \bigotimes_{i=1}^n \pi_i^* \mathcal{M}^1_C ),
\end{equation}
where $C^n$ is the Cartesian product of $n$ copies of $C$, and  the $\pi_i$, $i=1, \ldots, n$ are the projections on each individual factor. In local coordinates $z_i := z(p_i)$, $p_i \in C$, $i=1,\ldots,n$ a degree $n$ differential can be written as\footnote{For simplicity we will omit the tensor product symbol $\otimes$ between the differentials.}
\begin{equation}
W_n(p_1, \ldots, p_n) = w_n(z_1, \ldots, z_n) \mathrm{d} z_1 \cdots \mathrm{d} z_n,
\end{equation}
where $w(z_1, \ldots, z_n)$ is meromorphic in each variable.

To define the recursion we need to introduce a particular degree $2$ differential.
\begin{defn}\label{d:bergman}
We define $W_2^0(p_1,p_2)$ to be the \emph{fundamental normalized bi-differential}  \cite[p.20]{Fay:1970} which is uniquely defined by the conditions:
\begin{itemize}
\item It is symmetric, $W_2^0(p_1, p_2) = W_2^0(p_2,p_1)$;
\item It has its only pole, which is double, along the diagonal $p_1 = p_2$, with no residue; its expansion in this neighborhood has the form
\begin{equation}
W_2^0(p_1, p_2) = \left(\frac{1}{(z_1-z_2)^2} + \text{regular} \right) \mathrm{d} z_1 \mathrm{d} z_2;
\end{equation}
\item It is normalized by requiring that its periods about a basis of $A$-cycles on $C$ vanish.\footnote{$W_2^0(p_1,p_2)$ has also been called \emph{Bergman kernel} in the literature. It is the second order derivative of the $\log$ of the prime-form on $C$ \cite{Fay:1970}.}
\end{itemize}
\end{defn}

\begin{example}
For $C = \mathbb{P}^1$, $W_2^0$ is simply the Cauchy differentiation kernel:
\begin{equation}\label{eq:cauchy}
W_2^0(p_1, p_2) = \frac{\mathrm{d} z_1 \mathrm{d} z_2}{(z_1 - z_2)^2}.
\end{equation}
\end{example}

Having now defined the main ingredients, we can introduce the Eynard-Orantin recursion.

\subsubsection{The Eynard-Orantin topological recursion}

Let $\{ W^g_n \}$ be an infinite sequence of meromorphic differentials $W^g_n(p_1, \ldots, p_n) \in M^n_C$ for all integers $g\geq 0$ and $n > 0$ satisfying the condition $2g - 2 + n \geq 0$. We say that the differentials with $2g-2+n > 0$ are \emph{stable}; $W_2^0(p_1, p_2) \in M^2_C$ is the only unstable differential.

Let us introduce the shorthand notation $S = \{p_1, \ldots, p_{n} \}$. Then:
\begin{defn}\label{d:recursion}
We say that the meromorphic differentials $W_n^g$ satisfy the \emph{Eynard-Orantin topological recursion} if:
\begin{equation}
W^g_{n+1}(p_0, S) =  \sum_{\lambda=1}^n \underset{q=a_\lambda}{\text{Res}}  K_\lambda(p_0, q) \Big(W^{g-1}_{n+2} (q, s_\lambda(q), S) +\sum_{\substack{g_1 + g_2 = g \\ I \cup J = S}} W^{g_1}_{|I|+1}(q, I) W^{g_2}_{|J|+1}(s_\lambda(q), J) \Big),
\end{equation}
where $K_\lambda(p_0,q)$ is the \emph{Eynard kernel} defined below. The recursion here is on the integer $2g - 2 + n$, which is why it is called a topological recursion. The initial condition of the recursion is given by the unstable $W_2^0 \in M^2_C$ defined above.
\end{defn}

\begin{defn}\label{d:kernel}
The \emph{Eynard kernel} $K_\lambda(p_0,q)$ is defined, in local coordinate $q$ near $a_\lambda$, by
\begin{equation}\label{eq:eynard}
K_\lambda(p_0,q) = \frac{1}{2} \frac{\int_q^{s_\lambda(q)} W^0_2(p_0,q')}{\omega(q) - \omega(s_\lambda(q))},
\end{equation}
where $\omega(q)$ is the meromorphic one-form $\omega(q) = y(q) \mathrm{d} x(q)$. Here, $\frac{1}{\mathrm{d} x(q)}$ is the contraction operator with respect to the vector field $\left( \frac{\mathrm{d} x}{\mathrm{d} q} \right)^{-1} \frac{\partial}{\partial q}$.
\end{defn}

Definitions \ref{d:bergman}, \ref{d:recursion} and \ref{d:kernel} together define the Eynard-Orantin topological recursion for the curve $C$. 

\subsubsection{The $F_g$'s}

We can also extend the construction to $n = 0$ objects, $F_g := W_0^g$, which are just numbers. To construct the $F_g$, $g \geq 2$ (the stable ones), we need an auxiliary equation. Let us first define
\begin{equation}\label{eq:primitive}
\Phi(q) = \int_0^q \omega(q'),
\end{equation}
which is the primitive of the one-form $\omega(q)$ for an arbitrary base point $0$. We then define:
\begin{defn}\label{d:fg}
The numbers $F_g$, $g \geq 2$, are constructed from the one-forms $W_1^g(p)$ by:
\begin{equation}
F_g = \frac{(-1)^g}{2-2g} \sum_{\lambda=1}^n \underset{q=a_\lambda}{\text{Res}} \Phi(q) W_1^g(q).
\end{equation}
\end{defn}  

\begin{rem}
We note that in the definition of $F_g$ we introduced a factor $(-1)^g$ which is absent in the original formalism \cite{Eynard:2007}. This factor arises from the fact that the generating parameter $\lambda$ in Gromov-Witten theory, introduced in (\ref{Flambda}), is related to the generating parameter $\hbar$ considered in the formalism of \cite{Eynard:2007} as $\hbar=-i \lambda$. One way to see this identification arises from the fact that MacMahon function (\ref{eq:macmah}) should be expressed in terms of a unified parameter $q=e^{-\hbar}=e^{i\lambda}$.
\end{rem}

To summarize, given an affine curve $C$, the Eynard-Orantin topological recursion constructs an infinite tower of meromorphic differentials $W_n^g(p_1, \ldots, p_n)$ (definition \ref{d:recursion}), and numbers $F_g := W_0^g$ (definition \ref{d:fg}), for $g \geq 0$, $n >0$, satisfying the stability condition $2g-2+n > 0$. The recursion kernel is the Eynard kernel (definition \ref{d:kernel}), and the initial condition of the recursion is the fundamental normalized bi-differential on $C$ (definition \ref{d:bergman}).

The Eynard-Orantin topological recursion appeared in the realm of matrix models. It turns out that the large $N$ limit of Hermitian matrix models is encoded in an affine curve, known as the \emph{spectral curve} of the matrix model. If we choose $C$ to be this spectral curve, then the meromorphic differentials $W_n^g$ and numbers $F_g$ compute respectively the correlation functions and free energies of the matrix model in its large $N$ limit. This is where the recursion comes from. For this reason, we will often refer to the $W_n^g$ as \emph{correlation functions}, and to the $F_g$ as \emph{free energies}.

However, it has been realized that the same topological recursion has many applications in other areas of mathematics, in particular in Gromov-Witten theory and mirror symmetry, which is the content of the ``remodeling conjecture'' \cite{Bouchard:2009, Marino:2008} to which we now turn to.

\subsection{The remodeling conjecture}

The remodeling conjecture  \cite{Bouchard:2009, Marino:2008} is an application of the Eynard-Orantin recursion in the world of Gromov-Witten theory and mirror symmetry. Roughly speaking, the statement of the conjecture is that if we apply the Eynard-Orantin recursion to the complex curve $\Sigma$ mirror to a toric Calabi-Yau threefold $X$, the $W_g^n$ and $F_g$ constructed by the recursion compute respectively the open and closed genus $g$ amplitudes of B-model topological string theory on $\Sigma$, which are mirror to the open and closed generating functions of Gromov-Witten invariants of $X$.

\subsubsection{The Eynard-Orantin recursion on mirror curves}

Let us now be a little more precise. We explained how to construct the mirror curves $\Sigma$ in subsection \ref{s:mirror}. Now we want to apply the Eynard-Orantin topological recursion, described in subsection \ref{s:recursion} to these mirror curves.

The first thing to notice is that the mirror curves are slightly different from the affine curves used to define the Eynard-Orantin recursion. 

First, the mirror curves $\Sigma$ are families of curves, depending on the parameters $r_i$, $i=1, \ldots, k$. This means that the differentials $W^g_n$ will depend on the $r_i$, and the $F_g$ will also be functions of the $r_i$.

Second, the mirror curves are algebraic curves in $(\mathbb{C}^*)^2$ instead of $\mathbb{C}^2$. Nevertheless, we can still apply the Eynard-Orantin recursion, if we replace the one-form $\omega(q)$ by its $\mathbb{C}^*$ version $\omega(q) = \log y(q) \frac{\mathrm{d} x(q)}{x(q)}$.\footnote{Equivalently, we could work in exponential variables, in which case the one-form would remain $y \mathrm{d }x$.} This being said, the fact that the mirror curves are in $(\mathbb{C}^*)^2$ has important consequences for properties of the meromorphic differentials and free energies constructed through the recursion.

\subsubsection{Statement of the conjecture}    \label{sss:conjecture1}

We can now state the remodeling conjecture more precisely. Its statement could be split in two parts:
\begin{conjecture}[Remodeling conjecture  \cite{Bouchard:2009, Marino:2008}]
Let $\Sigma_f$ be the framed mirror curve to a toric Calabi-Yau threefold $X$. 
\begin{enumerate}
\item The free energies $F_g$ constructed by the Eynard-Orantin recursion are mapped by the mirror map to the genus $g$ generating functions of Gromov-Witten invariants of $X$.
\item The correlation functions $W_g^n$ are mapped by the open/closed mirror map to the generating functions of framed open Gromov-Witten invariants.
\end{enumerate}
\end{conjecture}

\begin{rem}
Note that we did not define open Gromov-Witten invariants in this paper; henceforth we will concentrate on the first part of the conjecture, involving the $F_g$ and the standard Gromov-Witten invariants.
\end{rem}

The conjecture has been tested computationally for many toric geometries (see for instance \cite{Bouchard:2009, Bouchard:2010, Brini:2009, Marino:2008}). For the simplest case $X=\mathbb{C}^3$, it has been proved that the correlation functions obtained through the recursion reproduce the topological vertex calculation \cite{Chen:2009, Zhou:2009, Zhou:2009ii}. A proof for the resolved conifold is provided in \cite{Eynard:2008ii} (although it does not address the constant terms). A general proof of the conjecture (aside from constant terms again) was outlined in \cite{Eynard:2010, Eynard:2010ii}, although a few gaps remain. The special case of Hurwitz numbers \cite{Bouchard:2009ii}, which can be seen as a consequence of the remodeling conjecture, has been proved in \cite{Borot:2009, Eynard:2009}.

\subsubsection{Constant maps}

\label{s:question}

However, one subtle point was not addressed carefully in the original conjecture and subsequent work. According to the remodeling conjecture, the free energies $F_g$ should be mapped to the generating functions of Gromov-Witten invariants. As we saw in subsection \ref{s:GW}, the leading term in the generating functions of Gromov-Witten invariants $F_g$ correspond to the contribution from constant maps with $\beta=0$, $N_{g,0}$, see \eqref{eq:constant}. Then one could ask:
\begin{question}
Is the remodeling conjecture true for the whole free energies $F_g$, including the constant terms, or just for the reduced free energies $F_g^{\beta \neq 0}$? In other words, are the $F_g$ constructed from the Eynard-Orantin recursion, when we set the parameters $r_i=0$, giving the same constant terms $N_{g,0}$ as in Gromov-Witten theory?
\end{question}

As far as we are aware, all the checks and proofs of the remodeling conjecture so far only involved non-constant terms, since they compared with the topological vertex, which only computes the reduced Gromov-Witten theory. In this paper we argue that the answer to the above question is a surprising yes:  \emph{the $F_g$ constructed by the recursion, miraculously, compute precisely the constant terms in Gromov-Witten theory}! This is surprising for various reasons, as mentioned in the introduction. But interestingly, it turns out that it is due to the particular geometry of mirror curves. Let us now see how it goes.

\section{Contributions for constant maps}

\subsection{The role of the ramification points}

\label{s:bpoints}

Before we address the computation of constant terms for general toric Calabi-Yau threefolds, we need to understand the geometry of mirror curves a little better, and the relation between this geometry and the form of the Eynard-Orantin recursion. One particular thing which is rather puzzling is the role of the ramification points of the $x$-map. Those are fundamental from the recursion point of view, but seem to come out of nowhere from a mirror symmetry standpoint. Let us now see how they can interpreted geometrically.

Recall that a mirror curve $\Sigma$ is an algebraic curve in $(\mathbb{C}^*)^2$; its associated Riemann surface has a certain number of punctures.
\begin{lem}\label{l:ram}
Let $\Sigma$ be a curve mirror to a toric Calabi-Yau threefold $X$. Let $n$ be the number of ramification points (which are all simple) of the $x$-map $x: \Sigma \to \mathbb{C}^*$. Assume that the $x$-map is a branched covering. Then
\begin{equation}
n = \chi(X) =  \text{number of torus fixed points of $X$}.
\end{equation}
In particular, the number of ramification points does not depend on the degree of the $x$-map.
\end{lem}

\begin{proof}
The lemma is a consequence of the Riemann-Hurwitz formula. Recall that given a branched covering $f: X \to Y$ between (not necessarily compact) Riemann surfaces, the Riemann-Hurwitz formula tells us that
\begin{equation}
\chi(X) = d \chi(Y) - b,
\end{equation}
where $d$ is the degree of $f$, $\chi(X)$ and $\chi(Y)$ are the topological Euler characteristics, and $b$ is the branching index. In our case, assuming that the map $x$ is a branched covering and that its ramification locus is composed of $n$ simple ramification points, we get
\begin{equation}
\chi(\Sigma) = d \chi(\mathbb{C}^*) - n.
\end{equation}
But $\chi(\mathbb{C}^*) = 2 - 2 g(\mathbb{C}^*) - 2 = 0$, hence
\begin{equation}
n = - \chi(\Sigma).
\end{equation}
 In particular, $n$ does not depend on the degree $d$ of the $x$-map.

Now any $\Sigma$ has a pair of pants decomposition, and $\chi(\Sigma)$ is equal to the number of pair of pants times the Euler characteristic of a pair of pants (a thrice punctured sphere), which is $-1$. But recall that $\Sigma$ can be seen as fattening the toric diagram of $X$ (see figure \ref{fig-toric}), hence each pair of pants correspond to a trivalent vertex of the toric diagram. Therefore, if we let $V$ be the number of vertices in the toric diagram of $X$, we obtain $\chi(\Sigma) = - V$. Hence $n=V =  \text{number of torus fixed points of $X$} = \chi(X)$.
\end{proof}

\begin{rem}
Note that this result is special to the case of curves in $(\mathbb{C}^*)^2$. For affine curves in $\mathbb{C}^2$ as originally considered by Eynard and Orantin, the number of ramification points of course depends on the degree of the map, as can be seen by applying Riemann-Hurwitz to that case.
\end{rem}

So what we have found is that the ramification points of the $x$-map are in one-to-one correspondence with the torus fixed points of $X$. Hence, taking residues at the ramification points in the $x$-map seems to be a direct B-model mirror the A-model localization of Gromov-Witten invariants with respect to the torus action. It would be very nice to make this analogy more precise.

Moreover, what we have seen is that the ramification points are also in one-to-one correspondence with the pair of pants in the decomposition of the Riemann surface $\Sigma$; there is one ramification point on each pair of pants. Therefore, it seems that there is a pair of pant decomposition built in the residue process of the topological recursion. At each step of the recursion, we are summing over residues at the ramification points; in other words, we are summing over contributions from each pair of pants. This also has an analog in the A-model; by the fattening prescription, the pair of pants decomposition of $\Sigma$ is equivalent to breaking the toric diagram along its internal edges. But in Gromov-Witten theory the topological vertex theory is precisely built by gluing $\mathbb{C}^3$ patches, which correspond to decomposing the toric diagram along internal edges. Hence the residue process can be seen as an analog decomposition on the B-model side, where  the $\mathbb{C}^3$ patch decomposition of $X$ is replaced by a pair of pant decomposition of $\Sigma$. We will make this remark more precise for constant terms in the following.

Finally, note that there is a very important assumption included in the statement of the lemma, which is that the $x$-map has to be a branched covering; otherwise one cannot use Riemann-Hurwitz directly. But is this assumption satisfied for mirror curves? Interestingly, we can show that it is satisfied for framed mirror curves $\Sigma_f$, but only \emph{for a generic choice of framing $f$}, as we prove below. However, there may exist a finite number of values of $f$ for which it is not satisfied. For these pathological choices of framing, what happens is that the map fails to be a branched covering at a finite number of points in $\mathbb{C}^*$ (for instance it may not be surjective). We can still apply Riemann-Hurwitz, but we first need to take care of these points either by removing them or plugging in some punctures in $\Sigma$; the result is then that the number of ramification points of the $x$-map becomes less than $\chi(X)$.

As a consequence, it turns out that, as we will see, for these pathological choices of framing the recursion \emph{does not} produce the right free energies $F_g$. This is rather striking, since, as explained in \cite{Bouchard:2009}, the framing transformation is a symplectic transformation (in the sense that it preserves the symplectic form $\frac{\mathrm{d} x}{x} \wedge \frac{\mathrm{d} y}{y}$ on $(\mathbb{C}^*)^2$), hence the $F_g$ should be invariant under framing transformations, according to \cite{Eynard:2007, Eynard:2007ii, Eynard:2008}. What happens is that if we compute the $F_g$ for the framed curve $\Sigma_f$, treating $f$ as a \emph{parameter}, then the resulting $F_g$ \emph{do not} depend on $f$. So they are invariant in this sense, but the calculation must be done treating $f$ as a parameter to get the right answer. As a result, framing becomes an essential part of the calculation, and can be understood as some sort of ``regularization'' procedure for the mirror curve. 

Let us now prove that the assumption is satisfied for framed mirror curves, for a generic choice of framing.

\begin{lem}
Let $\Sigma_f$ be a framed mirror curve. Then, for generic choice of framing $f$, the map $x: \Sigma_f \to \mathbb{C}^*$ is a branched covering. 
\end{lem}

\begin{proof}
To construct the mirror curve, we do the reparameterization
\begin{equation}
(x,y) = (X Y^{-f}, Y).
\end{equation}
The unframed curve always has the form
\begin{equation}
1 + x + y + \sum_{i=1}^k r_i x^{m_i} y^{n_i} = 0,
\end{equation}
for some integers $m_i$ and $n_i$. After the framing reparameterization, the curve becomes
\begin{equation}
1 + X Y^{-f} + Y + \sum_{i=1}^k r_i (X Y^{-f})^{m_i} Y^{n_i} = 0.
\end{equation}
Let $d$ be the degree of the map $X : \Sigma_f \to \mathbb{C}^*$. The map will be a branched covering if, after removing the set of branch points in $\mathbb{C}^*$ and the set of ramification points in $\Sigma_f$, the resulting map is a covering. That is, for all $p \in \mathbb{C}^*$ not in the branching locus, $X^{-1}(p) \in \Sigma_f$ consists in exactly $d$ points.

Generically, given a $X$ the curve becomes a degree $d$ polynomial in $Y$, hence the preimage generically consists in $d$ points in $\Sigma_f$. However, the $X$-map may fail to be a branched covering if for a given $f$, there exists a $X \in \mathbb{C}^*$ such that either the term of highest or lowest degree in $Y$ vanishes (in the former case, we lose a point in the preimage, while in the latter case one of the points in the preimage goes to $Y=0$ which is a puncture of $\Sigma_f$). This will happen if there are two highest or lowest degree terms in $Y$ with different powers of $X$. It is easy to enumerate the possible such ``bad'' choices of $f$, which depend on the particular integers $m_i$ and $n_i$. The point is that for any framed mirror curve, there is only a finite number of such bad $f$, hence the map is a branched cover for a generic $f$.
\end{proof}

Let us illustrate the role of framing with a simple example.

\begin{example}
\label{ex:c3framing}
As an example of this issue with framing, consider the framed curve mirror to $\mathbb{C}^3$, \eqref{eq:fmc3}:
\begin{equation}
\Sigma_f = \left \{H(X,Y):= X+ Y^f + Y^{f+1} = 0 \right \} \subset (\mathbb{C}^*)^2.
\end{equation}
The curve $\Sigma_f$ has three punctures. Hence $\chi(\Sigma_f) = 2 - 3 = -1$, and since $\chi(\mathbb{C}^3) = 1$, we indeed have that $\chi(\Sigma_f) = - \chi(\mathbb{C}^3)$.  We thus expect that the $X$-map should have a single ramification point. For generic $f$, the ramification point is given by the unique solution (in $\mathbb{C}^*$) of 
\begin{equation}
\frac{\partial H(X,Y)}{\partial Y} = Y^{f-1} (f + (f+1) Y) = 0,
\end{equation}
that is,
\begin{equation}
Y_* = - \frac{f}{f+1}, \qquad X_* = f^f (-1-f)^{-1-f}.
\end{equation}
However, this analysis fails for the choices of framing $f=0$ and $f=-1$, since the map $X$ is not a branched covering anymore. In the case of $f=0$, the curve becomes
\begin{equation}
X + Y + 1 = 0.   \label{C3-mirror1}
\end{equation}
The point $X=-1 \in \mathbb{C}^*$ has no preimage in $\Sigma_f$ under the $X$-map (the preimage would be $(-1,0)$, which is a puncture of $\Sigma_f$), hence the map is not surjective. Thus it is not a branched covering. It is a branched covering however from $\Sigma_f$ to $\mathbb{C}^* \setminus \{ -1 \}$ (of course, it is one-to-one), hence $X$ has no ramification point.

Similarly, for $f=-1$, the curve becomes
\begin{equation}
X Y + 1 + Y = 0.    \label{C3-mirror2}
\end{equation}
Again, the point $X = -1 \in \mathbb{C}^*$ has no preimage in $\Sigma_f$. So we lose a puncture again, and the $X$-map has no ramification point.

It turns out that the mirror $\mathbb{C}^3$ curves are most often encountered in the literature in the form (\ref{C3-mirror1}) or (\ref{C3-mirror2}), which correspond to these bad choices of framing. Similarly, typical representations of the conifold curve (\ref{curveConifold}), or other curves considered in the literature, often correspond to pathological choices of framing. As we discuss in sec. \ref{ss-symplectic}, for such bad choices of framing the topological recursion does not produce the correct free energies.
\end{example}

Now that we understand the role of the ramification points and the relation with the pair of pants decomposition of the mirror curve, let us move on to the study of constant terms.

\subsection{The simplest toric Calabi-Yau threefold: $\mathbb{C}^3$}

\label{s:c3}

We start with the simplest toric Calabi-Yau threefold, $X = \mathbb{C}^3$. In this case, the only contributions in Gromov-Witten theory are given by constant maps. From a mirror symmetry point of view, there is no K{\"a}hler parameters, hence there is no $r$-parameter in the mirror curve (that is, the mirror curve is really a curve, not a family of curves). So the question in subsection \ref{s:question} translates in this case into the question whether the $F_g$ computed by the recursion give precisely the genus $g$ Gromov-Witten invariants $N_{g,0}$ for constant maps to $\mathbb{C}^3$.

\subsubsection{Calculation of the low genus invariants}

Let us apply the topological recursion to the mirror curve of $\mathbb{C}^3$. As explained above, we need to used the regularized (or framed) mirror curve, given by \eqref{eq:fmc3}.
\be
\Sigma^{\mathbb{C}^3} = \{ H(X,Y) := X + Y^f + Y^{f+1} = 0 \}.   \label{HxyC3}
\ee
From (\ref{HxyC3}) we determine a dependence
\be
Y(X) = -1 + \sum_{k=1}^{\infty} (-1)^{k(f+1)} \frac{(k f + k - 2)!}{(k f - 1)! k!} X^k     \label{yx}
\ee

Solving the equation $\frac{\partial H(X,Y)}{\partial Y} = 0$ we find a single ramification point at:
\be
X_* = f^f (-1-f)^{-1-f},\qquad Y_* = - \frac{f}{1+f}.   \label{ystarC3}       
\ee

Let us introduce a local coordinate $p$ in the neighborhood of the branch point $Y_*$
\be
Y = Y_* + p,    \label{ystarp}
\ee
The conjugate point $\overline{Y} := s_*(Y)$ obtained from the deck transformation near the ramification point $a_*$, that is $X(Y)=X(\overline{Y})$, can be found in a series expansion
\begin{equation}
\overline{Y} = Y_* -p + \frac{2(f^2-1)p^2}{3f} -  \frac{4(f^2-1)^2 p^3}{9f^2} + \frac{2(1+f)^3 (-22 + 57 f - 57 f^2 + 22 f^3) p^4}{135 f^3} + \mathcal{O}(p^5).
\end{equation}
The primitive $\Phi(p)$ of the one-form $\omega = \log Y \frac{\mathrm{d} X}{ X}$, as defined in \eqref{eq:primitive}, can be determined in the exact form
\begin{equation}
\Phi(p) = \frac{f}{2} \Big(\log\big(p - \frac{f}{1+f}\big)\Big)^2 +  \log\Big(p - \frac{f}{1+f}\Big) \log\Big(\frac{1+p+fp}{1+f}\Big) + 
 \textrm{Li}_2\Big( \frac{f - p - f p}{1 + f}\Big).
\end{equation}
The curve (\ref{HxyC3}) has genus zero and so the fundamental normalized bi-differential is simply the Cauchy differentiation kernel \eqref{eq:cauchy}. The Eynard kernel \eqref{eq:eynard} is found in a series expansion as
\bea
K(p,q) & = & \Big(-\frac{f^2}{2 (1 + f)^4 p^2\,q} - \frac{f (f-1)}{2 (1 + f)^3 p^2} +    \label{kernel-q} \\
& & + \frac{f \big( 4(1+f)^2p^2 + 2(1-f^2)p -3f \big)\, q}{ 6 (1 + f)^4 p^4} + \mathcal{O}(q^2) \Big) \frac{\mathrm{d} p}{\mathrm{d} q}.
\eea

With all the above ingredients we can follow the recursion procedure. The first correlators which we find read
\begin{align}
W^1_1(p) =& -\frac{p^2+f^4 p^2+2 f^3 p (-1+2 p)+2 f p (1+2 p)+f^2 (-3+6 p^2)}{24(1+f)^4 p^4} \mathrm{d}p, \\
W^0_3(p_1,p_2,p_3) =& \frac{f^2}{(1+f)^4 p_1^2 p_2^2 p_3^2} \mathrm{d} p_1 \mathrm{d} p_2 \mathrm{d} p_3,
\end{align}
and the others are more complicated. Having found $W^2_1(p)$ we determine
\be
F_2 = -\frac{1}{2} \textrm{Res}_{p\to 0} \Phi(p) W^2_1(p) = \frac{1}{5760}.
\ee
Similarly, having found $W^3_1(p)$, we determine
\be
F_3 = \frac{1}{4} \textrm{Res}_{p\to 0} \Phi(p) W^3_1(p) = -\frac{1}{1451520}.
\ee
Continuing calculations (restricting to fixed framing $f=2$ for higher genera) we find
\be
F_4 =  \frac{1}{87091200} , \qquad \qquad F_5 = -\frac{1}{2554675200}  ,
\ee
\be
F_6 = \frac{691}{31384184832000} , \qquad \qquad F_7 = -\frac{691}{376610217984000} .
\ee

\subsubsection{Comparison and conjecture}

As we saw in subsection \ref{s:GW}, the free energies for $X = \mathbb{C}^3$, which are just given by the constant contributions $F_g = N_{g,0}$, satisfy:
\be
F_g = \frac{1}{2} (-1)^g \times \frac{|B_{2g}| |B_{2g-2}|}{2g(2g-2)(2g-2)!},
\ee
for $g \geq 2$. One can immediately check that the $F_g$ that we obtained above agree with this general formula. 

From these calculations we propose:
\begin{conjecture}\label{c:c3}
Let $\Sigma_f$ be the framed curve mirror to $X = \mathbb{C}^3$. Then the free energies obtained through the Eynard-Orantin recursion are given by:
\be
F_g = \frac{1}{2}(-1)^g \times \frac{|B_{2g}| |B_{2g-2}|}{2g(2g-2)(2g-2)!}.
\ee
\end{conjecture}
This is highly non-trivial. We cannot prove this conjecture at the moment, but we hope to report on it. If the conjecture is true, it would complete the proof of the full remodeling conjecture for the case of $X=\mathbb{C}^3$; the second part of the conjecture, dealing with the correlation functions, has already been proved by Chen and Zhou \cite{Chen:2009, Zhou:2009}. Note that it may be possible to prove this conjecture using very recent results of Eynard \cite{Eynard:2011}.

\subsection{Contributions for constant maps for all toric Calabi-Yau threefolds}

In this section we argue that for a general toric Calabi-Yau threefold $X$, 
the constant terms $N_{g,0}$ are equal to $n$ times the $F_g$ of the vertex, 
where $n$ is the number of ramification points. 
In other words, we conjecture that:
\begin{conjecture}\label{c:X}
Let $X$ be a toric Calabi-Yau threefold, and $\Sigma_f$ its framed mirror curve. Let $F_g^X$ be the free energies computed from the Eynard-Orantin recursion applied to $\Sigma_f$. Then
\begin{equation}
\lim_{r_i \to 0} F_g^X = n F_g^{\mathbb{C}^3},   \label{conjecture3}
\end{equation}
where the $F_g^{\mathbb{C}^3}$ are the free energies computed from the curve mirror to $\mathbb{C}^3$.
\end{conjecture}

Thus, for constant terms we obtain a pair of pants decomposition of the free energies, where each pair of pants contributes a single copy of $F_g^{\mathbb{C}^3}$.

Note that conjecture \ref{c:c3} and \ref{c:X} together imply that for all toric Calabi-Yau threefolds $X$, the constant terms in the $F_g$ computed by the recursion are the same ones as in Gromov-Witten theory (refer to  subsection \ref{s:known}).

We checked this conjecture computationally, at low genus, in various toric geometries.
In what follows we will also give a general argument supporting this conjecture, and then prove it in the simple case when $X$ is the resolved conifold. We plan to report on a proof for arbitrary toric geometry soon.

\subsubsection{Ramification points vs. $\mathbb{C}^3$ patches -- examples} 

In this subsection we illustrate, firstly, a correspondence between ramification points and torus fixed points (or equivalently toric vertices used to built a toric diagram). Secondly, we argue that, as a consequence of this fact, each $\mathbb{C}^3$ patch -- or pair of pants -- contributes a factor of $F_g^{\mathbb{C}^3}$ to the constant part of the free energy.

We recall that a toric diagram is dual to the Newton polygon, which can be read off from the equation of the mirror curve. After appropriate triangulation of the Newton polygon, it consists of triangular patches which are in one-to-one correspondence with trivalent vertices of the toric diagram. From lemma \ref{l:ram} proved in section \ref{s:bpoints}, these patches and vertices are also in one-to-one correspondence with ramification points of the mirror curve (for generic framing). Below we demonstrate how, upon relevant rescalings of variables in the mirror curve equation, and taking infinite limit of K{\"a}hler parameters, one can get rid of all but one ramification points, and reduce the curve to the local neighbourhood of such a single ramification point of his or her choice. In this limit the remnant of the original curve is just the mirror $\mathbb{C}^3$ curve, whose ramification point can be identified with the single ramification point which survived the above limit. This is similar to the subtlety described after lemma \ref{l:ram} in section \ref{s:bpoints}; however then losing a ramification point was a feature of a pathological choice of framing, while presently it is a natural consequence of the limit which we consider, and the fact that in this limit the corresponding parameters $Q$ are set to zero in the mirror curve equation. 

Let us consider now the behavior of constant contributions $Z^{\beta=0}$ to the partition function in the above limit. By definition these contributions do not change when parameters $Q$ are varied. However they will change discontinuousely when, in the above decoupling limit, some parameters $Q$ are set strictly to zero, so that we lose corresponding ramification points. As we argue below, we can take such a decoupling limit in various ways, in order to focus on local neighoburhood of a chosen ramification point. As each such neighbourhood can be identified with the mirror $\mathbb{C}^3$ curve, its contribution to the total partition function will be that of $F_g^{\mathbb{C}^3}$. If we take a decoupling limit in which such a ramification point is not lost, it will still contribute $F_g^{\mathbb{C}^3}$. If we take a limit so that we lose this point, the corresponding $F_g^{\mathbb{C}^3}$ factor will drop out of the computation. We can also start with a single $\mathbb{C}^3$ mirror curve, whose free energy is $F_g^{\mathbb{C}^3}$, and consider the opposite process, in which we turn on certain (combination of) $Q$ parameters, and build up more general curve, by bringing in additional ramification points to the picture. Each such new ramification point will contribute $F_g^{\mathbb{C}^3}$ to the free energy of the more general curve. In consequence, if the mirror curve we wish to consider has $n$ ramification points, in total they will contribute $n F_g^{\mathbb{C}^3}$ to the partition function, in agreement with (\ref{conjecture3}). 

Instead of being general, we illustrate now how this mechanism works in some examples. To start with we note that even for a single $\mathbb{C}^3$ patch, whose toric diagram involves a single trivalent vertex and the Newton polygon consists of a single triangular patch, the form of the corresponding single branch point depends on a particular form of the mirror curve. For example, for the curve $H(x,y)=1+x+y=0$ given in (\ref{curveC3}), the Newton polygon and the dual toric diagram are shown in the left panel in fig. \ref{fig-C3}. For another parametrization $H'(x,y)=x+y+x y=0$ the Newton polygon and the toric diagram are shown in the right panel. Regularizing these curves by introducing framing $x\to X Y^f$ we get respectively\footnote{Note that this corresponds to the framing $-f$ in the definition (\ref{def-framing}).}
\be
H_{-f}(X,Y) = 1 + Y + XY^f, \qquad \qquad H'_{-f}(X,Y) = Y + XY^f + XY^{f+1},   \label{HfC3}
\ee
and the corresponding ramification points are characterized by
\be
Y_{*} = \frac{f}{1-f},\qquad \qquad Y'_{*} = \frac{1-f}{f}.   \label{brptsC3}
\ee
These two curves $H=0$ and $H'=0$ are related by a (multiplicative) symplectic transformation $(X,Y)\mapsto(X^{-1},Y^{-1})$. Therefore their free energies should be the same and equal to $F_g^{\mathbb{C}^3}$; we also checked that this is the case by explicit computations.

\begin{figure}[htb]
\begin{center}
\includegraphics[width=0.6\textwidth]{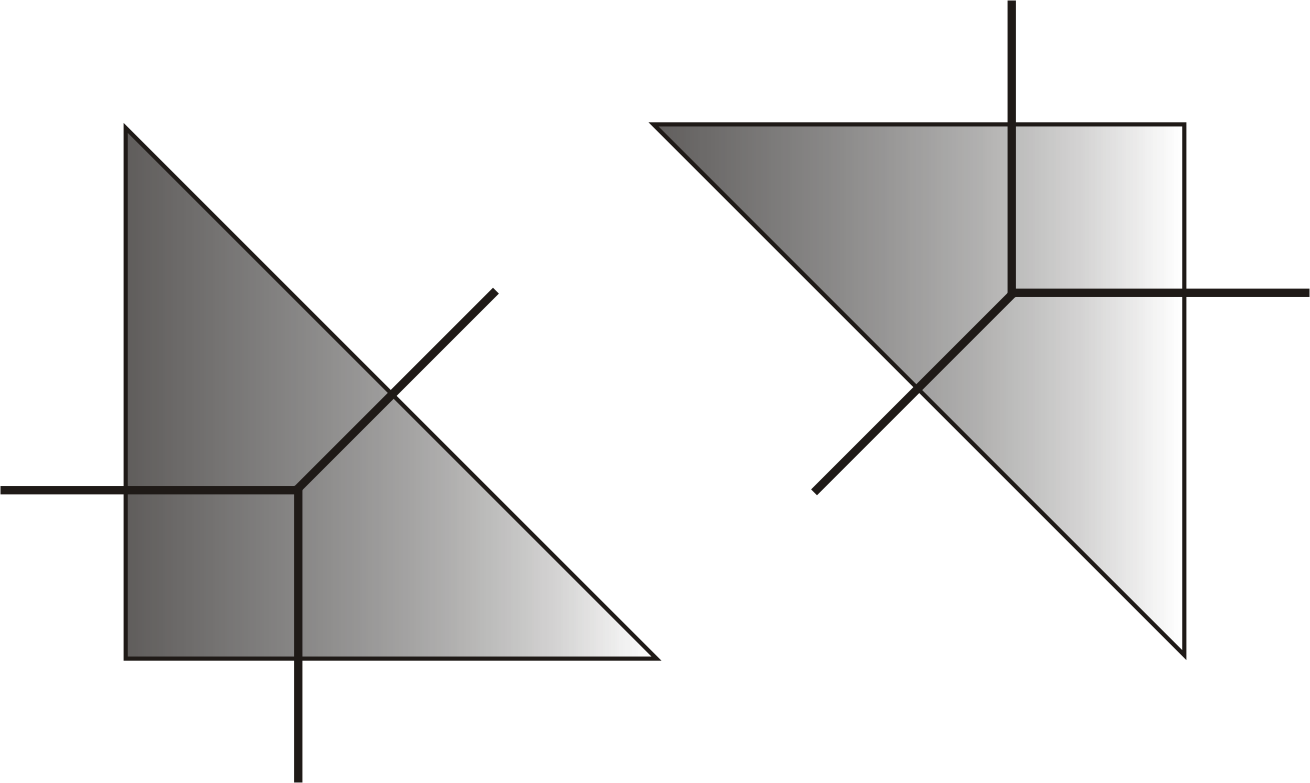} 
\begin{quote}
\caption{\emph{Toric diagrams for $\mathbb{C}^3$, corresponding to $H(x,y)$ and $H'(x,y)$ discussed in the text, involve a single trivalent vertex. There is a single corresponding branch point, whose regularized form, for the two forms of $\mathbb{C}^3$ in our example, is presented in (\ref{brptsC3}).}   } \label{fig-C3}
\end{quote}
\end{center}
\end{figure}

Let us discuss now the behavior of ramification points in case of the conifold. We will also consider two forms of such curve, however related by a very simple symplectic transformation, involving just a multiplication of $X$ or $Y$ by a constant. Therefore, even though the equations for two such curves are different, there is an obvious one-to-one correspondence between their ramification points. To start with, we introduce the notation $X=r x y^{-1}, Y=y$ and $Q=r^{-1}$ for the curve in (\ref{curveConifold}), so that the mirror curve equation takes the form $\widetilde{H}(X,Y)=1 + X + Y + QXY = 0.$ The Newton polygon consists of two triangular patches, and the dual toric diagram involves two trivalent vertices, as shown in fig. \ref{fig-conifold}. Regularizing this curve by introducing framing $X\to XY^f$ we get 
\begin{equation}
\widetilde{H}_{-f}(X,Y) = 1 + Y + X Y^f + Q X Y^{f + 1},
\end{equation}
and there are two ramification points, 
\begin{equation}
\widetilde{Y}_{*,\pm} = \frac{1 - f - Q - f Q \pm \sqrt{-4 f^2 Q + (-1 + f + Q + f Q)^2}}{2 f Q},
\end{equation}
whose expansion in small $Q$ takes form
\be
\widetilde{Y}_{*,+} = \frac{f}{1-f} + \frac{fQ}{(1-f)^3} + \mathcal{O}(Q^2), \qquad 
\widetilde{Y}_{*,-} = \frac{1-f}{f Q} + \frac{1}{f(f-1)} + \mathcal{O}(Q^2). \label{brptsConi}
\ee
We see that for $Q\to 0$, the curve $\widetilde{H}_{-f}(X,Y)$ reduces to the $\mathbb{C}^3$ case $H_{-f}(X,Y)$ in (\ref{HfC3}), and the ramification point $\widetilde{Y}_{*,+}$ reduces to the $\mathbb{C}^3$ ramification point $Y_{*}$ in (\ref{brptsC3}), while $\widetilde{Y}_{*,-}$ runs away from the picture. In this limit only one trivalent vertex survives, and the corresponding dual triangular patch, shown in grey in the left panel in fig. (\ref{fig-conifold}), reduces to the $\mathbb{C}^3$ patch in the left panel in fig. (\ref{fig-C3}). In consequence, after this limit has been taken, the remaining constant contribution to the free energy will be just the $F_g^{\mathbb{C}^3}$  of the corresponding ramification point $\widetilde{Y}_{*,+}$. 

\begin{figure}[htb]
\begin{center}
\includegraphics[width=0.8\textwidth]{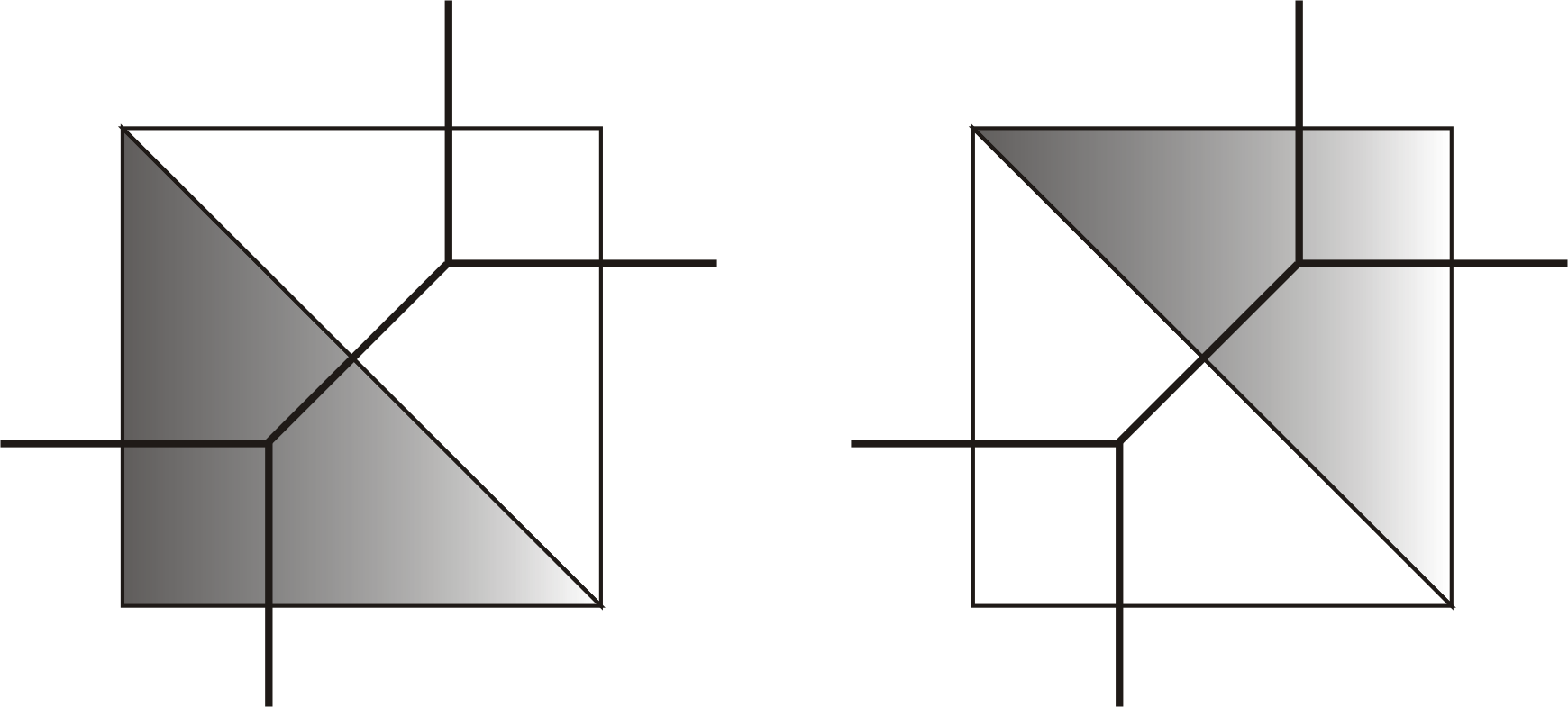} 
\begin{quote}
\caption{\emph{Toric diagram for the conifold. By appropriate reparametrization of $X$ and $Y$ in the mirror curve equation, and subsequently sending the K{\"a}hler parameter $Q\to 0$, we can focus on either $\mathbb{C}^3$ patch highlighted in grey, reproducing one of the $\mathbb{C}^3$ patches in fig. \ref{fig-C3}. One of the two ramification points of the conifold in (\ref{brptsConi}) or (\ref{brptsConi2}) runs away in $Q\to 0$ limit, while the remaining ramification point reduces to the relevant $\mathbb{C}^3$ branch point in (\ref{brptsC3}).} }\label{fig-conifold}
\end{quote}
\end{center}
\end{figure}

We can also focus on the other patch of the conifold digram, shown in grey in the right panel in fig. (\ref{fig-conifold}). To achieve that we introduce rescaled coordinates $X'=QX, Y'=QY$, so that $\widetilde{H}(X,Y)=1+X+Y+QXY\mapsto \widetilde{H}'(X',Y')=Q+X'+Y'+X'Y'$.
Regularizing this resulting curve by introducing framing we get
\begin{equation}
\widetilde{H}'_{-f}(X',Y') = Q + Y' + X' Y'^f + X' Y'^{f + 1},
\end{equation}
and the two corresponding ramification points read now
\begin{equation}
\widetilde{Y}'_{*,\pm} = \frac{1 - f - Q - f Q \pm \sqrt{-4 f^2 Q + (-1 + f + Q + f Q)^2}}{2 f}.
\end{equation}
In the small $Q$ limit the curve $\widetilde{H}'_{-f}(X,Y)$ reduces to $H'_{-f}(X,Y)$ in (\ref{HfC3}), while the ramification points have the expansion
\be
\widetilde{Y}'_{*,+} = \frac{f Q}{1 - f} + \frac{f Q^2}{(1 - f)^3} + \mathcal{O}(Q^2), \qquad 
\widetilde{Y}'_{*,-} = \frac{1-f}{f} + \frac{Q}{(f-1) f} + \mathcal{O}(Q^2), \label{brptsConi2}
\ee
and now it is $\widetilde{Y}'_{*,-}$ which reduces to $Y'_{*}$ in (\ref{brptsC3}), while $\widetilde{Y}'_{*,+}$ runs away. Correspondingly, the surviving, highlighted part of the conifold toric diagram in the right side of fig. \ref{fig-conifold} reduces to the $\mathbb{C}^3$ patch in the right panel in fig. \ref{fig-C3}. In this representation, after the decoupling limit has been taken, the remaining contribution to the free energy will also be that of $F_g^{\mathbb{C}^3}$ for the remaining ramification point $\widetilde{Y}'_{*,-}$.

We therefore showed that in appropriate decoupling limits the conifold curve reduces to some version of the $\mathbb{C}^3$ curve, and correspondingly each of its two ramification points contributes one factor of $F_g^{\mathbb{C}^3}$. Therefore such contributions must be present simultaneously if the decoupling limits is not taken, and the constant contributions of the conifold mirror curve must be equal to $2F_g^{\mathbb{C}^3}$, in agreement with (\ref{conjecture3}). For the conifold we prove that this is indeed the case in the next section.

\begin{figure}[htb]
\begin{center}
\includegraphics[width=0.4\textwidth]{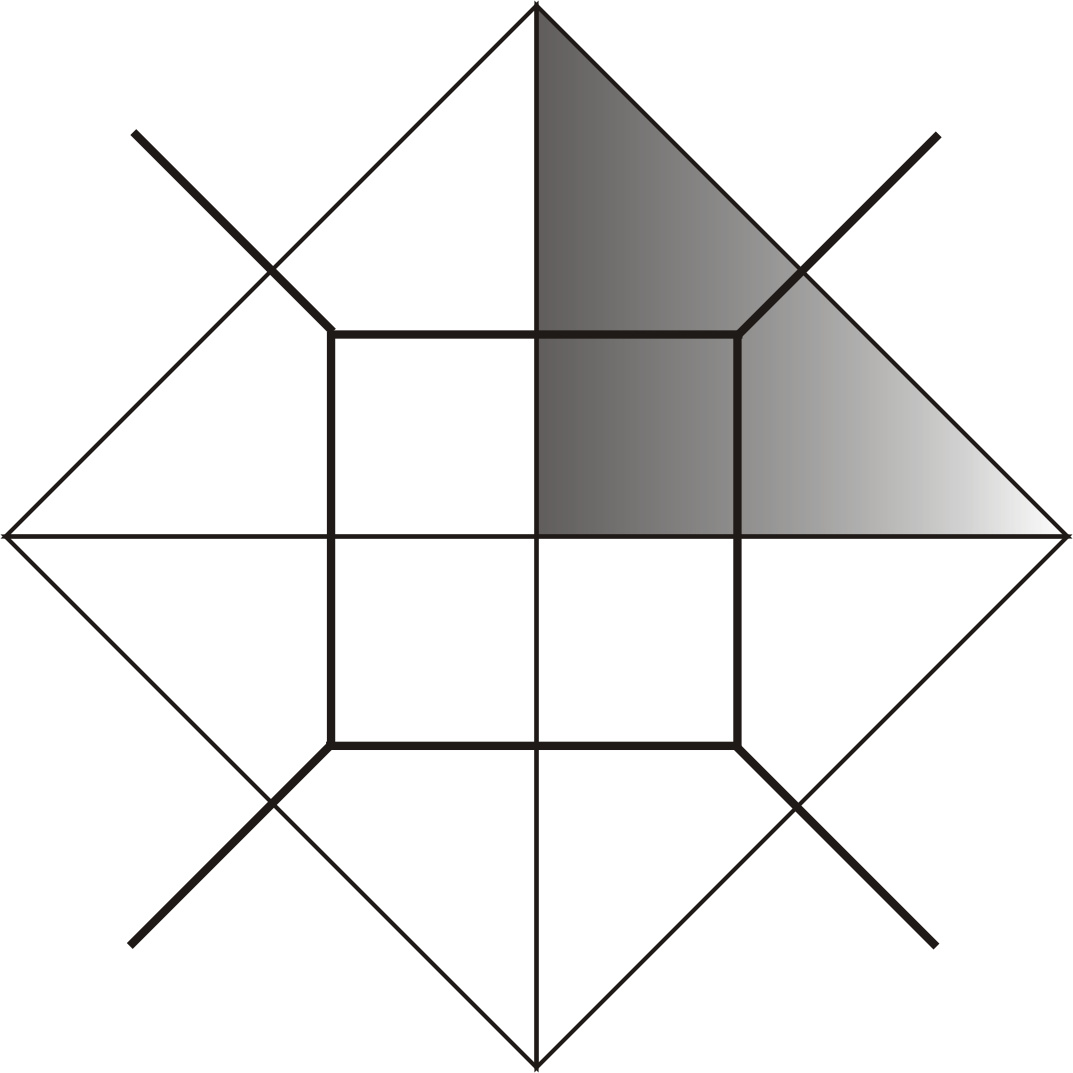} 
\begin{quote}
\caption{\emph{Toric diagram for local $\mathbb{P}^1\times \mathbb{P}^1$ geometry, which reduces to the $\mathbb{C}^3$ (corresponding to the grey triangle patch) in the limit of vanishing K{\"a}hler parameters.}} \label{fig-P1P1}
\end{quote}
\end{center}
\end{figure}

It is also clear that the above mechanism can be generalized to arbitrary toric manifolds. For example, the local $\mathbb{P}^1\times \mathbb{P}^1$ geometry is encoded in the mirror curve
\be
\widehat{H}_{-f}(X,Y) = 1 + Y + X Y^f + \frac{Q_1}{X Y^f} + \frac{Q_2}{Y},   \label{HP1P1}
\ee
with $Q_1$ and $Q_2$ representing K{\"a}hler parameters of the two independent two-cycles. There are 4 ramification points, and the corresponding Newton polygon and toric diagram, shown (for $f=0$) in fig. \ref{fig-P1P1}, consist respectively of 4 triangles and 4 trivalent vertices. By appropriate rescaling of $X$ and $Y$ and subsequently taking the limit $Q_1,Q_2\to 0$, one can focus on any of the 4 patches of $\mathbb{C}^3$. In particular, taking $Q_1,Q_2\to 0$ directly in (\ref{HP1P1}), the curve (\ref{HP1P1}) reduces to $H_f(X,Y)$ in (\ref{HfC3}). At the same time we focus on the patch highlighted in grey in fig. \ref{fig-P1P1}, which is nothing but the patch in the left panel in fig. \ref{fig-C3}, and the surviving ramification point reduces to $Y_*$ in (\ref{brptsC3}). This ramification point contributes one factor of $F_g^{\mathbb{C}^3}$ to the partition function. Considering various decoupling limits, analogously as in the conifold case, we can focus on the neighbourhood of each of the 4 ramification points of $\widehat{H}_f$, and in consequence the total constant part of the free energy will be given by $4 F_g^{\mathbb{C}^3}$. 

For a general toric geometry, whose toric diagram consists of $n$ vertices and the corresponding mirror curve has (for generic framing) $n$ ramification points, the total constant part of the free energy would be equal to $n F_g^{\mathbb{C}^3}$, in agreement with (\ref{conjecture3}).

\subsubsection{Proof for the resolved conifold}   \label{sss:conifold}

Let us now prove conjecture \ref{c:X} when $X$ is the resolved conifold.
\begin{lem}\label{l:conifold}
Let $X= \mathcal{O}_{\mathbb{P}^1} (-1) \oplus  \mathcal{O}_{\mathbb{P}^1} (-1) $ be the resolved conifold, and $\Sigma_f$ its framed mirror curve. Then
\begin{equation}
\lim_{r\to 0} F_g^X = 2 F_g^{\mathbb{C}^3}.
\end{equation}
\end{lem}

\begin{proof}
The framed curve mirror $\Sigma_f$ to the conifold is
\begin{equation}
H_f(X,Y):= 1 + X Y^{-f} + Y + r X Y^{-f-1} = 0.
\end{equation}
It can be parameterized by
\begin{equation}
Y = t, \qquad X = - \frac{t^{f+1}(1+t)}{r+t}.
\end{equation}
This is a genus 0 curve with four punctures. It has two ramification points $a_0, a_1 \in \Sigma_f$ at
\begin{align}
a_0=& \frac{-2 r-f (r+1)+\sqrt{r-1} \sqrt{(r-1) f^2+4 r f+4 r}}{2 (f+1)}, \\
a_1 =& -\frac{2 r+f (r+1)+\sqrt{r-1} \sqrt{(r-1) f^2+4 r f+4 r}}{2 (f+1)}.
\end{align}

We want to compute the $F_g$ by
\begin{equation}
F_g = \frac{(-1)^g}{2-2g} \sum_{\lambda=0}^1 \underset{t=a_\lambda}{\text{Res}} \Phi(t) W_1^g(t),
\end{equation}
and take the $r \to 0$ limit. 

Note that in the $r \to 0$ limit, we obtain that $a_1 \to 0$ and $a_0 \to - \frac{f}{f+1}$, the latter being the ramification point of the curve mirror to $\mathbb{C}^3$. It is also clear that in the limit $r \to 0$, the curve becomes just the curve mirror to $\mathbb{C}^3$ (compare with \eqref{HxyC3}). Thus, it is clear that
\begin{equation}\label{eq:a0}
\lim_{r \to 0} \left(  \frac{(-1)^g}{2-2g} \underset{t=a_0}{\text{Res}} \Phi(t) W_1^g(t)\right) = F_g^{\mathbb{C}^3}.
\end{equation}
The question is what happens for the residue at the other ramification point $a_1$, since at $r \to 0$ we have that $a_1 \to 0$, which is a puncture of the Riemann surface (hence $\Phi(t)$ is not well defined there because of logarithms).

But it turns out that the curve $\Sigma_f$ has a nice symmetry which comes to the rescue. Consider reparameterizing the curve by 
\begin{equation}
w = \frac{a_0 a_1}{t} = \frac{r}{t},
\end{equation}
where we used the fact that $a_0 a_1 = r$. This reparameterization clearly exchanges the two ramification points, \emph{i.e.}, $t = a_1$ corresponds to $w = a_0$ and vice-versa. Using this reparameterization we can write
\begin{equation}
  \frac{(-1)^g}{2-2g} \underset{t=a_1}{\text{Res}} \Phi(t) W_1^g(t) =  \frac{(-1)^g}{2-2g} \underset{w=a_0}{\text{Res}} \Phi \left(  \frac{a_0 a_1}{w} \right) W_1^g \left(  \frac{a_0 a_1}{w} \right) .
\end{equation}
But it turns out that the curve behaves very nicely under this reparameterization. Indeed, we have that
\begin{align}
Y(t) =& Y \left(  \frac{a_0 a_1}{w} \right)  = r Y^{-1}(w),\\
X(t) =& X \left(  \frac{a_0 a_1}{w} \right) = r^f X^{-1}(w).
\end{align}
But this is just a transformation that preserves the symplectic form
\begin{equation}
\frac{\mathrm{d} X}{X} \wedge \frac{\mathrm{d} Y}{Y}.
\end{equation}
Moreover, it preserves the ramification points of $X$, hence the deck transformation, etc. So the recursion kernel is the same in both $t$ and $w$ coordinates, thus, the correlation functions should be the same. What this means is that we have, for the one-point correlation functions,
\begin{equation}
W^g_1(t) = W^g_1 \left( \frac{a_0 a_1}{w} \right) = W^g_1(w).
\end{equation}
It is easy to check computationally that this is indeed satisfied.

Moreover, we know that the one-form $\omega$ transforms as:
\begin{align}
\omega(t) =& \log Y(t) \frac{\mathrm{d} X(t)}{X(t)} \\
=& \log Y \left(  \frac{a_0 a_1}{w} \right) \frac{\mathrm{d} X \left( \frac{a_0 a_1}{w} \right)}{X \left( \frac{a_0 a_1}{w} \right)} \\
=& - \left(\log(r) - \log Y(w) \right) \frac{\mathrm{d} X(w)}{X(w)}\\
=& \omega(w) - \log(r)  \frac{\mathrm{d} X(w)}{X(w)}.
\end{align}
Thus we obtain that
\begin{align}
\Phi(t) =& \Phi  \left( \frac{a_0 a_1}{w} \right) \\
=& \int_0^{a_0 a_1/w} \omega(z) \\
=& \int_\infty^w \omega \left( \frac{a_0 a_1}{u} \right) \\
=& \int_\infty^w \left(\omega(u) - \log(r)  \frac{\mathrm{d} X(u)}{X(u)}  \right) \\
=& \Phi(w) + \text{cst} - \log(r) \log X(w).
\end{align}
Putting all this together, we obtain
\begin{align}
  \frac{(-1)^g}{2-2g} \underset{t=a_1}{\text{Res}} \Phi(t) W_1^g(t) =& \frac{(-1)^g}{2-2g} \underset{w=a_0}{\text{Res}} \Phi \left(  \frac{a_0 a_1}{w} \right) W_1^g \left(  \frac{a_0 a_1}{w} \right) \\
=&  \frac{(-1)^g}{2-2g} \underset{w=a_0}{\text{Res}} \left( \Phi(w) + \text{cst} - \log(r) \log X(w) \right) W^g_1(w).
\end{align}
$W^g_1(w)$ has no residue at $a_0$, hence the residue of the term $(\text{cst}) W^g_1(w)$ vanishes. It is also easy to prove (see for instance \cite{Eynard:2007, Eynard:2008, Norbury:2010}) that
\begin{equation}
 \underset{w=a_0}{\text{Res}} \left( \log X(w) \right) W^g_1(w) = 0.
\end{equation}
Therefore, we get that
\begin{equation}
 \frac{(-1)^g}{2-2g} \underset{t=a_1}{\text{Res}} \Phi(t) W_1^g(t) =  \frac{(-1)^g}{2-2g} \underset{w=a_0}{\text{Res}} \Phi(w) W^g_1(w).
\end{equation}
This means that the two residues in the calculation of $F_g$ contribute exactly the same term to $F_g$, that is,
\begin{equation}
F_g = 2  \frac{(-1)^g}{2-2g} \underset{t=a_0}{\text{Res}} \Phi(t) W_1^g(t).
\end{equation}
Hence from \eqref{eq:a0} we conclude that
\begin{equation}
\lim_{r \to 0} F_g = 2 F_g^{\mathbb{C}^3}.
\end{equation}
\end{proof}

\begin{rem}
Note that what we proved for the conifold is much stronger than lemma \ref{l:conifold}. We proved that for every $F_g$, the two pairs of pants in the decomposition of the curve $\Sigma_f$ (\emph{i.e.} the two ramification points) give precisely the same contribution to the full $F_g$, not just for constant terms. Hence there really is a pair of pant decomposition here, at the level of the full free energy. However, this is probably an artefact of the nice symmetry of the curve $\Sigma_f$; we expect this decomposition to hold only for constant terms for a general toric Calabi-Yau threefold.  Note also that this lemma completes the proof of the full remodeling conjecture for the conifold, since the reduced part was proved in \cite{Eynard:2008ii}.
\end{rem}

\subsection{Beyond constant maps -- remodeling conjecture for the total free energy}

At this point it is clear that a proper understanding of the remodeling conjecture (given in section \ref{sss:conjecture1}) should involve both constant contributions, as well as non-constant contributions. In other words, our results strongly support the claim, that the partition function computed by the topological recursion reproduces the full Gromov-Witten partition function given in (\ref{Ztotal}). This claim follows by combining tests for non-constant contributions already present in the literature, which we mentioned in the introduction, with our results for constant contributions. It is instructive to illustrate this claim in at least one example. Let us therefore consider the total term $F_2$ for the resolved conifold.

Regarding constant contributions to $F_2$, our results, in particular the proof in section \ref{sss:conifold}, assert that they are given by twice the contribution (\ref{eq:FP}) (with $g=2$). On the other hand, it is known that the non-constant part of the free energies for the conifold, for $g\geq 2$, are given by polylogarithms:
\begin{equation}
F_g^{\beta\neq 0} = (-1)^{g + 1} \frac{B_{2g}}{2g (2g-2)!}  \textrm{Li}_{3-2g}(Q).
\end{equation}
Therefore the total genus two Gromov-Witten free energy for the conifold reads
\begin{equation}
F_2^{GW} = \frac{1}{2880} + \frac{1}{240} \textrm{Li}_{-1} (Q).
\end{equation}
On the other hand an explicit application of the topological recursion for the conifold mirror curve results in
\begin{equation}
F_2 = \frac{1 + 10Q + Q^2}{2880 (1-Q)^2}.
\end{equation}
Both these results indeed agree
\begin{equation}
F^{GW}_2 = F_2.
\end{equation}

\section{Discussion}

In this paper we studied the free energies constructed from the Eynard-Orantin recursion for mirror curves. We showed that there is a pair of pants decomposition built in the recursion, through the residue calculation at the ramification points. We argued that the constant terms computed from the recursion agree with the constant terms obtained in Gromov-Witten theory. This follows from the pair of pant decomposition of the mirror curves. In consequence, we extended the scope of the remodeling conjecture to the full free energies.

This work however opens up a certain number of questions. Here are a few:

\subsection{Limit theorem}

In \cite[section 8]{Eynard:2007}, limits of families of curves are also studied. Roughly speaking, what they obtain is that the recursive construction of the $F_g$ commutes with the limit on the curve. Their setup is however different from ours; they consider a limit where a branch point becomes singular.

In the case studied in this paper, what we obtained is that the limit $r_i \to 0$ does not commute with the recursion. Indeed, if $\Sigma^X$ is the curve mirror to $X$,
\begin{equation}
\lim_{r_i \to 0} F_g[\Sigma^X] \neq  F_g \left[ \lim_{r_i \to 0} \Sigma^X \right],
\end{equation}
by which we mean that the $r_i \to 0$ limit of the $F_g$ of a given mirror curve $\Sigma^X$ are not the same as the $F_g$ of the curve with $r_i=0$. More precisely, it is easy to see from the definition that for any mirror curve, after setting $r_i=0$ the curve $\Sigma^X$ reduces to the curve $\Sigma^{\mathbb{C}^3}$ mirror to $\mathbb{C}^3$. So $F_g  \left[ \lim_{r_i \to 0} \Sigma \right] = F_g[\Sigma^{\mathbb{C}^3}]$. But what we argued in the previous section is that
\begin{equation}
\lim_{r_i \to 0} F_g[\Sigma^X] = n F_g[\Sigma^{\mathbb{C}^3}],
\end{equation}
that is, there is an all important multiplicative factor of $n$, which is the number of pair of pants in $\Sigma^X$.

So in our context, the recursive construction does not commute with the limit $r_i \to 0$. It would be interesting to compare more precisely with the limit theorem of \cite{Eynard:2007}, even though the setups are different. It would also be interesting to see whether the kind of limit statement that we obtained is a property of mirror curves, or whether it holds for general curves in $\mathbb{C}^*$, using the relation with pair of pants decompositions.

\subsection{Symplectic invariance}   \label{ss-symplectic}

In \cite{Eynard:2007, Eynard:2007ii, Eynard:2008} it is stated that the free energies $F_g$ constructed from the recursion are invariant under the group of transformations of the maps $x$ and $y$ that preserve the symplectic form $\mathrm{d} x \wedge \mathrm{d} y$ (in the $\mathbb{C}*$ context, the symplectic form is $\frac{\mathrm{d} x}{x} \wedge \frac{\mathrm{d} y}{y}$). 

From the form of the recursion, invariance is clear for any transformation which does not modify the ramification points of the $x$-map. The tricky part of the statement is for transformations that do modify the ramification points, in particular transformations that \emph{change the number of ramification points}. Invariance under these transformations was apparently shown in \cite{Eynard:2007ii}.

However, as we argued in subsection \ref{s:bpoints}, from our study of framing it became clear that the $F_g$ are not strictly invariant under framing transformations, which is a particular type of symplectic transformations. If we treat $f$ as a parameter, then one can show that the $F_g$ obtained from the recursion do not depend on $f$. However, there exists specific choices of framing for which the recursion does not produce the same $F_g$. This is because some framing transformations change the number of ramification points, hence accordingly change the constant terms of the $F_g$. A simple example was shown in example \ref{ex:c3framing}; for the choices of framing $f=0,-1$, the $X$-map has simply no ramification points, hence produces $F_g$ that are trivially zero, while for a generic choice of framing the $F_g$ are surely non-zero.

As a result, we argued that we should treat the framing $f$ as a parameter, and see it as some sort of regularization procedure. In this sense, ``symplectic invariance'' is somewhat restored since the $F_g$ do not depend on $f$, but strictly speaking the recursion produces different results for different choices of framing.

It is then worth asking whether this is a consequence of the fact that mirror curves are curves in $(\mathbb{C}^*)^2$, while the curves initially considered by Eynard and Orantin were affine curves, and the proof of \cite{Eynard:2007ii} may only apply to the latter. Let us study this question further.

\subsubsection{A counterexample}

It turns out that a similar phenomenon that we encountered with framing occur for much simpler curves, such as standard affine curves in $\mathbb{C}^2$. Let us study a very simple counterexample to the statement that the $F_g$ are invariant under symplectic transformations. 

Consider the affine curve
\begin{equation}
H(x,y) = y^2 - y x + 1 = 0.
\end{equation}
This curve has been studied for instance in \cite{Norbury:2009} (see also \cite[section 5.1]{Norbury:2010}). The $x$-map has two ramification points at $(2,1)$ and $(-2,-1)$. The $F_g$ can be computed; they are the orbifold Euler characteristics of the moduli spaces of genus $g$ curves \cite[Theorem 1]{Norbury:2009}:
\begin{equation}
F_g = \chi(\mathcal{M}_g).
\end{equation}
In particular, they are clearly non-zero.

Now consider the symplectic transformation $(x,y) \mapsto (x',y') = (y, -x)$, which preserves $\mathrm{d} x \wedge \mathrm{d} y = \mathrm{d} x' \wedge \mathrm{d} y'$. The curve becomes
\begin{equation}
G(x',y') = x'^2 + x' y' + 1 =0.
\end{equation}
The $x'$-map has of course no ramification points, it is a one-to-one map. Hence if we apply the recursion to this new $x'$-map, the free energies $F_g$ vanish identically! Clearly, they are not invariant under this particular symplectic transformation, since they were non-vanishing in the original parameterization. The problem is that the transformation changes the number of ramification points.

\subsubsection{SPP geometry}

It should be noted that this problem with symplectic transformations does not appear only when the newly parameterized curve has no ramification points. It also appears when the number of ramification points changes but remains non-zero. 

For instance, consider the mirror curve to the SPP geometry studied in \cite{Ooguri:2010} in sec. 3.2, which is of the form $Q_1 y^2 + x y + x + (1+Q_1 Q_2)y + Q_2 = 0$.
The $x$-map from this mirror curve has two ramification points. We computed the $F_g$ from this curve, and extracted the constant terms. What we obtained is that the constant terms $N_{g,0}$ are twice the $F_g$ of $\mathbb{C}^3$, consistent with the fact that there are two ramification points.

However, according to the general formulae in Gromov-Witten theory \eqref{eq:Fgrelation}, the constant terms $N_{g,0}$ should be three times the $F_g$ of $\mathbb{C}^3$, since the Euler characteristic of the SPP geometry is $\chi(X) = 3$.

What happens is that the choice of framing in which the curve is presented in \cite{Ooguri:2010} is one of these pathological choices of framing. It turns out that the general framed curve has indeed three ramification points, the same number as $\chi(X)$, as expected. Thus, for the framed curve we get the right constant terms $N_{g,0}$, which indeed do not depend on the framing $f$. But to get the right answer one first needs to regularize the curve by introducing framing; otherwise we do not get the right answer.

\begin{rem}
Indeed, for mirror curves, for any choice of parameterization such that the number of ramification points is not equal to $\chi(X)$, we will not get the right constant terms, from our results of the previous sections. But after regularizing the curve by introducing framing, the number of ramification points will be equal to $\chi(X)$, and the constant terms will match the expected numbers $N_{g,0}$ from Gromov-Witten theory.
\end{rem}

\subsubsection{Modify the recursion?}

Of course, symplectic invariance is one of the most desirable property of the $F_g$. But as we saw, there exists some symplectic transformations such that the newly parameterized curve does not produce the same invariants than the old curve. The examples that we gave all involve transformations that change the number of ramification points. It would be important to clarify which subgroup of transformations precisely are problematic.

But what would be even nicer is to modify the formulation of the recursion in order to restore explicit symplectic invariance. That is, rewrite the recursion such that for \emph{any} two curves related by a symplectic transformation, the $F_g$ computed by the recursion are equal. Perhaps what one needs to do is take into account explicitly the ramification points of both the $x$ and $y$ map simultaneously, or projectivize the curve in $\mathbb{P}^2$ in order to work in a compact setup.


\begin{thebibliography}{99}
\bibliographystyle{plain}

\bibitem{Aganagic:2004}
  M.~Aganagic, A.~Klemm, M.~Marino and C.~Vafa,
  ``Matrix model as a mirror of Chern-Simons theory,''
  JHEP {\bf 0402}, 010 (2004)
  [arXiv:hep-th/0211098].


\bibitem{Aganagic:2005}
  M.~Aganagic, A.~Klemm, M.~Mari\~no and C.~Vafa,
  ``The Topological Vertex,''
  Commun.\ Math.\ Phys.\  {\bf 254}, 425-478 (2005)
  [arXiv:hep-th/0305132].

\bibitem{Behrend:2005}
K.~Berhrend and B.~Fantechi,
``Symmetric obstruction theories and Hilbert schemes of points on threefolds,''
arXiv:math/0512556v1 [math.AG].

\bibitem{Borot:2009}
G.~Borot, B.~Eynard, M.~Mulase and B.~Safnuk,
``A matrix model for simple Hurwitz numbers, and topological recursion,''
arXiv:0906.1206v1 [math-ph].

\bibitem{Bouchard:2009}
  V.~Bouchard, A.~Klemm, M.~Mari\~no and S.~Pasquetti,
  ``Remodeling the B-model,''
  Commun.\ Math.\ Phys.\  {\bf 287}, 117 (2009)
  [arXiv:0709.1453 [hep-th]].

\bibitem{Bouchard:2010}
  V.~Bouchard, A.~Klemm, M.~Marino and S.~Pasquetti,
  ``Topological open strings on orbifolds,''
  Commun.\ Math.\ Phys.\  {\bf 296}, 589 (2010)
  [arXiv:0807.0597 [hep-th]].

\bibitem{Bouchard:2009ii}
V.~Bouchard and M.~Mari\~no,
``Hurwitz numbers, matrix models and enumerative geometry,'' 
in \emph{From Hodge Theory to Integrability and tQFT: tt*-geometry,} Proceedings of Symposia in Pure Mathematics, AMS (2008)	[arXiv:0709.1458v2 [math.AG]].



\bibitem{Brini:2009}
  A.~Brini and A.~Tanzini,
  ``Exact results for topological strings on resolved Y**p,q singularities,''
  Commun.\ Math.\ Phys.\  {\bf 289}, 205 (2009)
  [arXiv:0804.2598 [hep-th]].

\bibitem{Chen:2009}
  L.~Chen,
  ``Bouchard-Klemm-Marino-Pasquetti Conjecture for C**3,''
  arXiv:0910.3739 [math.AG].

\bibitem{Eynard:2008ii}                                         
B.~Eynard,
``All orders asymptotic expansion of large partitions,''
J.\ Stat.\ Mech. P07023 (2008).
[arXiv:0804.0381v2 [math-ph]].

\bibitem{Eynard:2011}
B.~Eynard,
``Intersection numbers of spectral curves,''
	arXiv:1104.0176v2 [math-ph].

\bibitem{Eynard:2010}
  B.~Eynard, A.~K.~Kashani-Poor and O.~Marchal,
  ``A matrix model for the topological string I: Deriving the matrix model,''
  arXiv:1003.1737 [hep-th].

\bibitem{Eynard:2010ii}
  B.~Eynard, A.~K.~Kashani-Poor and O.~Marchal,
  ``A Matrix model for the topological string II. The spectral curve and mirror
  geometry,''
  arXiv:1007.2194 [hep-th].

\bibitem{Eynard:2009}
B.~Eynard, M.~Mulase and B.~Safnuk,
``The Laplace transform of the cut-and-join equation and the Bouchard-Marino conjecture on Hurwitz numbers,''
arXiv:0907.5224v3 [math.AG].

\bibitem{Eynard:2007}
B.~Eynard and N.~Orantin,
``Invariants of algebraic curves and topological expansion,''
Comm.\ Numb.\ Theor.\ Phys.\ {\bf 1}, 347-452 (2007)
[arXiv:math-ph/0702045v4].

\bibitem{Eynard:2007ii}
B.~Eynard and N.~Orantin,
``Topological expansion of mixed correlations in the hermitian 2 Matrix Model and x-y symmetry of the $F_g$ invariants,''
	arXiv:0705.0958v1 [math-ph].

\bibitem{Eynard:2008}
B.~Eynard and N.~Orantin,
``Algebraic methods in random matrices and enumerative geometry,''
	arXiv:0811.3531v1 [math-ph].



\bibitem{Faber:2000}
C.~Faber and R.~Pandharipande, 
``Hodge integrals and Gromov-Witten theory,''
Invent.\ Math. {\bf 139}, 174-199 (2000)
[arXiv:math/9810173v1 [math.AG]].

\bibitem{Fay:1970}
J.~Fay, 
``Theta Functions on Riemann Surfaces,'' 
Lecture Notes in Mathematics, {\bf 352}, Springer–Verlag (1970).


\bibitem{Gopakumar:1998}
R.~Gopakumar and C.~Vafa, ``M-theory and topological strings I,''
arXiv:hep-th/9809187.

\bibitem{Hori:2000}
  K.~Hori and C.~Vafa,
 ``Mirror symmetry,''
  arXiv:hep-th/0002222.

\bibitem{SW-matrix}
A. Klemm and P. Su{\l}kowski,
``Seiberg-Witten theory and matrix models,''
Nucl. Phys. {\bf B819}, 400 (2009) 
[0810.4944 [hep-th]].

\bibitem{Levine:2008}
M.~Levine and R.~Pandharipande,
``Algebraic cobordism revisited,''
Invent.\ Math.\ {\bf 176}, 63-130 (2008)
[arXiv:math/0605196v1 [math.AG]]

\bibitem{Li:2006}
J.~Li,
``Zero dimensional Donaldson-Thomas invariants of threefolds,''
Geom.\ Topol.\ {\bf 10}, 2117-2171 (2006)
 [arXiv:math/0604490v2 [math.AG]].


\bibitem{Li:2009}
  J.~Li, C.~-C.~M.~Liu, K.~Liu and J.~Zhou,
  ``A Mathematical theory of the topological vertex,''
  Geom.\ Topol.\  {\bf 13}, 527-621 (2009)
  [arXiv:math/0408426 [math.AG]].

\bibitem{Marino:2004}         
  M.~Mari\~no,
  ``Chern-Simons theory, matrix integrals, and perturbative three manifold
  invariants,''
  Commun.\ Math.\ Phys.\  {\bf 253}, 25 (2004)
  [arXiv:hep-th/0207096].


\bibitem{Marino:2008}
  M.~Mari\~no,
  ``Open string amplitudes and large order behavior in topological string theory,''
  JHEP {\bf 0803}, 060 (2008)
  [arXiv:hep-th/0612127].

\bibitem{Marino:1999}
M.~Mari\~no and G.~Moore,
``Counting higher genus curves in a Calabi-Yau manifold,''
	Nucl.\ Phys.\ {\bf B543}, 592-614 (1999)
[arXiv:hep-th/9808131v1].




\bibitem{Maulik:2006}
D.~Maulik, N.~Nekrasov, A.~Okounkov and R.~Pandharipande, 
``Gromov-Witten theory and Donaldson-Thomas theory. I,'' 
Compos.\ Math. {\bf 142}, 1263–1285 (2006)
[arXiv:math/0312059v3 [math.AG]]

\bibitem{Maulik:2006ii}
D.~Maulik, N.~Nekrasov, A.~Okounkov and R.~Pandharipande,
``Gromov-Witten theory and Donaldson-Thomas theory. II,'' 
Compos.\ Math. {\bf 142}, 1286–1304 (2006)
[arXiv:math/0406092v2 [math.AG]]

\bibitem{Maulik:2011}
D.~Maulik, A.~Oblomkov, A.~Okounkov and R.~Pandharipande,
 ``Gromov-Witten/Donaldson-Thomas correspondence for toric 3-folds,''
Invent.\ Math., 1-45 (2011)
[arXiv:0809.3976v1 [math.AG]].

\bibitem{Norbury:2009}
P.~Norbury,
``String and dilaton equations for counting lattice points in the moduli space of curves,''
	arXiv:0905.4141v2 [math.AG].

\bibitem{Norbury:2010}
P.~Norbury and N.~Scott,
``Polynomials representing Eynard-Orantin invariants,''
arXiv:1001.0449v1 [math.AG].

\bibitem{matrix2star}
P. Su{\l}kowski,
\emph{Matrix models for 2* theories},
Phys. Rev. {\bf D80} 086006, (2009)  [0904.3064 [hep-th]].

\bibitem{Ooguri:2010}
H.~Ooguri, P.~Su{\l}kowski and M.~Yamazaki,          
``Wall Crossing As Seen By Matrix Models,''
Commun. Math. Phys. (2011)	
	[arXiv:1005.1293v1 [hep-th]].

\bibitem{Zhou:2009}
  J.~Zhou,
  ``Local Mirror Symmetry for One-Legged Topological Vertex,''
  arXiv:0910.4320 [math.AG].

\bibitem{Zhou:2009ii}
  J.~Zhou,
  ``Local Mirror Symmetry for the Topological Vertex,''
  arXiv:0911.2343 [math.AG].


\end{thebibliography}
\end{document}